\documentclass[11pt, leqno, letterpaper]{amsart}

\usepackage{graphicx}    

\usepackage[margin=1.2in]{geometry}

\usepackage{amsmath,amsthm,amsfonts,amssymb}

\newtheorem{theorem}{Theorem}

\newtheorem{proposition}[theorem]{Proposition}

\newtheorem{remark}[theorem]{Remark}

\begin{document}

\title[Discrete sums of GBM, Annuities and Asian Options]
{Discrete Sums of Geometric Brownian Motions, Annuities and Asian Options}

\author{Dan Pirjol}
\email{
dpirjol@gmail.com}

\author{Lingjiong Zhu}
\address
{Department of Mathematics \newline
\indent Florida State University \newline
\indent 1017 Academic Way \newline
\indent Tallahassee, FL-32306 \newline
\indent United States of America}
\email{
zhu@math.fsu.edu}

\date{20 May 2016}


\begin{abstract}
The discrete sum of geometric Brownian motions plays an important role 
in modeling stochastic annuities in insurance.
It also plays a pivotal role in the pricing of Asian options in mathematical 
finance. 
In this paper, we study the probability distributions of the infinite 
sum of geometric Brownian motions, the sum of geometric Brownian
motions with geometric stopping time, and the finite sum of the 
geometric Brownian motions. These results are extended to the
discrete sum of the exponential L\'evy process.
We derive tail asymptotics and compute numerically the asymptotic 
distribution function. We compare the results against the known 
results for the continuous time integral of the geometric Brownian motion up to an
exponentially distributed time. The results are illustrated with numerical
examples for life annuities with discrete payments, and Asian options.
\end{abstract}

\maketitle

\section{Introduction}

The valuation and risk management of annuities are important topics
in the actuarial science. There is a wide variety of contractual annuity payoffs,
ranging from fixed payouts to variable annuities, possibly with guaranteed 
benefits features. The pricing and valuation of such contracts have been considered in the
actuarial literature, under different choices of equity price, mortality and 
interest rate models \cite{FV,Hardy,VSHZ,DeS,RSST}. Most of the theoretical
work on variable annuities in the literature is in a continuous-time setting,
although in practice these instruments are defined and simulated in discrete
time. We note that discrete time models have been also considered in the 
actuarial literature \cite{GSY}. In this paper we will work in a discrete time
setting, and will compare our results with the continuous time approximation.

We will consider in this paper the case of a fixed payout annuity under the
geometric Brownian motion model for the equity returns. The interest rates will
be assumed to be constant and deterministic, and the mortality will be described
by a geometric distribution. This setting is of interest both from a theoretical
point of view, and as a simple starting point for building more complicated
and realistic models. The geometric Brownian motion model for equities prices
(the Black-Scholes model) is the simplest model available in the literature. 
The lognormal model (Independent Lognormal ILN) is also one of the 
recommended models in the American Academy of Actuaries report, and is used 
for generating pre-packaged economic scenarios (see \cite{AAA}, Appendix 2.). 
This model can be extended by adding stochastic volatility, modeled either 
as a diffusion or using a regime-switching approach.

Under the lognormal returns equity model, the present value of a fixed
coupon annuity can be related to the sum of geometric Brownian motions (GBM)
sampled on a discrete time points. The distributional properties of this
quantity, both in a discrete and continuous time setting, have been widely
studied in the mathematical finance and actuarial literature, see
\cite{DufresneReview,DufresneAnnuity} for overviews. 
Consider for example a stochastic annuity which pays a constant coupon $C$
at regular times with time step size $\tau$. Furthermore, assume that the
payments are made from a portfolio containing an asset with value process 
$U_t$ which is stochastic and follows a geometric Brownian motion:
\begin{equation}
U_t = e^{\sigma W_t + (m - \frac12 \sigma^2) t},
\end{equation}
where $W_t$ is a standard Brownian motion starting from $0$ at time $0$, and 
$\sigma, m$ are real parameters. $U_t$ represents the value at time $t$ of one 
unit of currency invested at time zero into the asset (for example stocks),
which has normally distributed log-returns.

The present value of the annuity, assumed to pay $n$ coupons at times $t_i$, is
\begin{equation}
S_n = \sum_{i=1}^n D_{t_i} \frac{C}{U_{t_i}},
\end{equation}
where $D_{t_i}$ are discount factors. This quantity is a random variable
and represents the liability, or amount of currency required at time zero
in order to be able to pay the annuity cash flows.
We are interested in the shape of the probability density of $S_n$,
and especially in the tails of this distribution, which give the probabilities
of extreme values of the liability represented by the annuity payments.
If the discounting rate $r$ is constant
and deterministic, the discount factors are exponential $D_{t_i} = e^{-r t_i}$
and the annuity value $S_n$ is proportional to the sum of a geometric
Brownian motion sampled on discrete times $t_i$
\begin{equation}\label{Sinv}
\tilde S_n = \sum_{i=1}^n e^{-\sigma W_{t_i} + (\frac12 \sigma^2 - m) t_i},
\end{equation}
where we redefined $m + r \to m$. 

A related quantity appears in the problem of pricing Asian options with
discrete sampling, which are derivatives linked to the average of the price
of an asset $S_t$
\begin{equation}\label{Avdef}
A_n = \frac{1}{n} \sum_{i=1}^n S_{t_i},
\end{equation}
under the assumption that $S_t$ follows a geometric Brownian motion
\begin{equation}
S_t = S_0 e^{\sigma W_t + (m-\frac12\sigma^2) t}.
\end{equation}
Here $W_{t}$ is a standard Brownian motion, and $\sigma,m$ are the volatility
and drift of the asset.

If the averaging times are uniformly distributed and the time step is
sufficiently small, the time average (\ref{Avdef}) can be approximated 
by the continuous time average
\begin{equation}
A_n = \frac{1}{t_n} \int_0^{t_n} S_t dt\,.
\end{equation}
This reduces the problem to the
study of the distributional properties of the time integral of the geometric
Brownian motion, which has been extensively studied in the literature
\cite{Dufresne1990,DufresneReview,Yor1,Yor2}. 

The distributional properties of the time-integral of the geometric
Brownian motion simplify very much in the infinitely large time limit. 
Define
\begin{equation}
Y_T = \int_0^T dt e^{\sigma W_t + (m-\frac12\sigma^2)t }.
\end{equation}
The following result was proven in \cite{Dufresne1990}, see Section 3
in \cite{DufresneLN} for a survey of related results.

\begin{theorem}[\cite{Dufresne1990,DufresneLN}]\label{Thm:Dufresne}
The limit $\lim_{T\to \infty} Y_T = Y_\infty$ exists in distribution if and only if
$m - \frac12\sigma^2 < 0$, and,
\begin{equation}
\frac{2}{\sigma^2 Y_\infty} = 
\mathbf{Gamma}\left( 1 - \frac{2m}{\sigma^2}, 1\right)\,.
\end{equation}
\end{theorem}

The limit probability density of $Y_\infty$ is given by the inverse Gamma 
distribution. The probability density function  $\mathbb{P}(Y_\infty\in
(z,z+dz)) = \phi_\infty(z)dz$, is given by
\begin{equation}\label{invGamma}
\phi_\infty(z;\sigma,m) = \left( \frac{2}{\sigma^2} \right)^{1-\frac{2m}{\sigma^2}}\left(\frac{1}{z^{2}}\right)^{1-\frac{m}{\sigma^{2}}}
\frac{1}{\Gamma\left(1-\frac{2m}{\sigma^2}\right)}
\exp\left( - \frac{2}{\sigma^2 z}\right)
\,.
\end{equation}
The cumulative distribution is
\begin{equation}\label{invGamma2}
\Phi_\infty(x;\sigma,m) = \mathbb{P}(Y_\infty < x) = 
\frac{1}{\Gamma\left(1-\frac{2m}{\sigma^2}\right)}
\Gamma\left(1 - \frac{2m}{\sigma^2}; \frac{2}{\sigma^2 x}\right),
\end{equation}
where $\Gamma(a,x) = \int_x^\infty  t^{a-1} e^{-t}dt$ is the incomplete 
Gamma function. 

We study in this paper the distributional properties of the discrete sum
of the GBM sampled on uniformly spaced times $t_i = i\tau$ 
\begin{equation}\label{Xndiscrete}
X_n = \sum_{i=1}^n e^{\sigma W_{t_i} + (m -\frac12\sigma^2 ) t_i}.
\end{equation}
The sum of the inverse of the GBM (\ref{Sinv})
can be obtained from this by replacing $\sigma \to -\sigma, m \to -m + \sigma^2$.
We will study the distributional properties of $X_n$, of the infinite sum of GBM
$X_\infty = \lim_{n\to \infty} X_n$, and of $X_N$, the sum of GBM with 
geometrically distributed $N$. We study the existence of the $n\to \infty$ 
limit, the tail asymptotics for all three cases, and the moments of the 
corresponding random variables, and also the convergence
of the discrete sums to the continuous time integrals. 
The results will then be applied to study annuities and Asian options.
We also also illustrate our results numerically.

The problem of determining the distribution of the discrete sum of GBM is
closely related to the study of the distribution of the sum of log-normals 
which has a long history, see \cite{reviewLN} for a literature overview. 
The simplest approach is to approximate the distribution with a log-normal 
random variable \cite{Levy,DufresneLN}. The tails asymptotics of the
sum of correlated log-normals has been studied in \cite{ARN,GXY,GT}.
Dufresne \cite{DufresneLN} studied the distribution of the sum 
(\ref{Xndiscrete}), generalized by the introduction of an weight function,
in the small volatility limit $\sigma \to 0$, and derived a limit theorem
for this quantity. 

Milevsky and Posner \cite{MP} studied the pricing of Asian options with discrete
monitoring. They applied the distributional properties of the infinite time 
integral of the GBM to the discrete sum $X_n$, and proposed to
use the inverse Gamma distribution as an approximation for the finite sum,
as a parametric approximation for pricing Asian options. The inverse Gamma
distribution is obtained only for continuous time; in discrete time the 
distribution is different, and we study here the corrections to the 
approximation of the discrete infinite sum with the continuous time integral.

The present authors \cite{PZ} studied the distributional properties of the sum 
(\ref{Xndiscrete})
in the $n\to \infty$ limit at fixed $\beta = \frac12 \sigma^2 \tau n^2$, and found
almost sure limit, fluctuations and large deviations of the average
$A_n = \frac{1}{n} X_n$ and then obtained the asymptotics for the prices
of out-of-the-money, in-the-money and at-the-money Asian options. 
Notice that the fixed $\beta$ and $n\rightarrow\infty$ limit corresponds
to the small maturity or small volatility regimes. For Asian options, 
the typical maturity can be $1$ year or $2$ year and the volatility is 
usually less than $100\%$. Thus the small maturity or small volatility regimes 
are of practical interest in business applications.
The case of fixed $\sigma, \tau$ and $n\rightarrow\infty$ gives 
the infinite sum of geometric Brownian motions $X_{\infty}$. 
This situation corresponds to the large maturity or large time-horizon regime, 
which are of practical interest for stochastic annuities, but less relevant 
for the pricing of Asian options.

The paper is organized as follows. In Section \ref{InfiniteSection}, 
we study the properties of the distribution of the infinite sum of a 
geometric Brownian motion.
We obtain a stochastic recurrence equation and use it to derive
an integral equation for the probability density function of the infinite sum
of the geometric Brownian motion $X_\infty$ sampled on a uniformly spaced
time grid with time step $\tau$. We derive the right and left
tail asymptotics for $X_\infty$. 
As the time step goes to zero $\tau \to 0$, we show that the infinite sum
of geometric Brownian motions converges in $L^1$ norm 
to the infinite integral, which is known to follow an inverse Gamma distribution
\cite{Dufresne1990}, see Theorem~\ref{Thm:Dufresne}.
In Section \ref{GeometricSection}, we repeat the analysis for the sum of 
geometric Brownian motions stopped at a geometric stopping time. 
In the context of the stochastic annuities, this corresponds to a
geometrically distributed mortality time. 
As the time step goes to zero, we show that this sum converges in distribution
to the continuous time integral of geometric Brownian motion with exponentially 
distributed terminal time, whose properties are well understood in the literature.
In Section \ref{LevySection}, we extend the results for the sums of geometric 
Brownian motions to the sums of exponential L\'{e}vy processes,
and in Section \ref{Sec:PositiveMoments}, we derive explicit expressions
for the positive moments of the infinite sum of GBM. 
In Section \ref{ApplicationSection}, we apply our results to study 
annuities with finite mortality, stochastic mortality
and the risk measures of the annuities, including the Value-at-Risk, and also
the pricing of discrete time Asian options.
We conclude the paper with numerical illustrations of the results in 
Section \ref{NumericalSection}.

\section{Infinite Sum of the GBM}\label{InfiniteSection}

We study in this section the infinite sum of GBM $X_\infty = 
\lim_{n\to \infty} X_n$, where $X_n$ is defined in (\ref{Xndiscrete}).
The density of this random variable
satisfies a functional relation, given by the following result.

\begin{proposition} 
The infinite sum of GBM $X_{\infty}$, if exists, satisfies the relation
\begin{equation}\label{Xrel}
X_\infty = \mathcal{A} (1 + X_\infty),
\end{equation}
where the equality is in distribution, and we denoted
\begin{equation}\label{Adef}
\mathcal{A} := e^{\sigma \sqrt{\tau} Z + (m-\frac12 \sigma^2)\tau},
\end{equation}
with $Z\sim N(0,1)$ independent of $X_{\infty}$.
\end{proposition}

\begin{proof}
The infinite sum of GBM, if exists, can be written as the sum of products 
of i.i.d. factors of the form (\ref{Adef})
\begin{equation}
X_\infty = \mathcal{A}_1 + \mathcal{A}_1 \mathcal{A}_2 
         + \mathcal{A}_1 \mathcal{A}_2 \mathcal{A}_3 + \cdots 
    = \mathcal{A}_1 (1 + \mathcal{A}_2 + \mathcal{A}_2 \mathcal{A}_3 + 
\cdots) = \mathcal{A}_1 (1 + X_\infty),
\end{equation}
where the last equality is in distribution. Since $\mathcal{A}_{i}$ are 
i.i.d., 
$\mathcal{A}_{1}$ is independent of $(1+\mathcal{A}_{2}+\mathcal{A}_{2}\mathcal{A}_{3}+\cdots)$ and
hence $\mathcal{A}$ is independent of $X_{\infty}$ in \eqref{Xrel}.
\end{proof}

\begin{remark}
An easier way to understand \eqref{Xrel} is by writing \eqref{Xrel}
as $X_{\infty}= \mathcal{A}(1+\hat{X}_{\infty})$ in distribution, where $X_{\infty}=\hat{X}_{\infty}$
in distribution and $\mathcal{A}$ is independent of $\hat{X}_{\infty}$.
\end{remark}

We study now the existence and uniqueness of the solution of the 
functional relation (\ref{Xrel}). 
This can be related to the existence of a limit distribution for the linear
recursion
\begin{equation}\label{Kesten}
x_{i+1} = a_i x_i + b_i
\end{equation}
where $(a_i,b_i)$ is an i.i.d. pair of random variables with values in
$[0,\infty) \times \mathbb{R}$. 
This reduces to the functional relation (\ref{Xrel}) by taking 
$a_i = b_i = \mathcal{A}$.

Linear random recursions of the form (\ref{Kesten}) have been widely studied 
in the probability literature \cite{Kesten,Vervaat,Dufresne1,MST,PZ2}.
Explicit solutions for the limit distribution in particular cases have been
found in \cite{Alsmeyer,Dufresne1990,Dufresne1,Goldie,GP}. 

The conditions for the existence and uniqueness of the limit distribution of
the recursion (\ref{Kesten}) are well-known 
\cite{Kesten,Vervaat}. The following are sufficient conditions 
for the existence and uniqueness of the limit distribution
\begin{equation}\label{cond12}
\mathbb{E}[\log a_i] <0 \,,\qquad \mathbb{E}[\log(b_i)_+] < \infty\,.
\end{equation}
This gives us the following result.

\begin{proposition}\label{prop:2}
The functional relation (\ref{Xrel}) has a solution, which is furthermore
unique, provided that the following inequality holds
\begin{equation}\label{ineq}
m < \frac12 \sigma^2 .
\end{equation}
\end{proposition}

\begin{proof}
The first condition (\ref{cond12}) gives
\begin{equation}
\mathbb{E}[\log \mathcal{A}] = \left(m-\frac12 \sigma^2\right) \tau < 0 \,,
\end{equation}
which is satisfied for $ m< \frac12\sigma^2$. 
The second condition (\ref{cond12}) is always satisfied. 

We conclude that the functional relation (\ref{Xrel}) has a solution which
is furthermore unique, provided that the inequality (\ref{ineq}) holds. We
note also that the condition under which this result holds is the same
as that in the Theorem~\ref{Thm:Dufresne} for the continuous time case.
\end{proof}

The relation (\ref{Xrel}) gives an integral equation for the 
probability density function of $X_\infty$, defined as 
$\mathbb{P}(X_\infty \in (x,x+dx)) = f(x;\beta,\rho) dx$. 
We will show that the probability density function of $X_{\infty}$ 
depends only on the parameters
\begin{equation}
\beta := \sigma^2 \tau \,,\qquad \rho := m\tau\,.
\end{equation}

\begin{proposition}
The function $f(x;\beta,\rho)$ satisfies the integral equation
\begin{equation}\label{integraleq}
f(x;\beta,\rho) = \frac{1}{x} \int_0^\infty \frac{dy}{\sqrt{2\pi \beta}} 
\exp\left( -\frac{1}{2\beta} 
\left[ \log \left(\frac{x}{1+y}\right) + \frac12 \beta - \rho \right]^2
\right)
f(y;\beta,\rho)\,.
\end{equation}
\end{proposition}

\begin{proof}
For any $x>0$, 
\begin{equation}
\mathbb{P}(X_{\infty}\leq x)=\mathbb{P}(\mathcal{A}(1+X_{\infty})\leq x)
=\int_{-\infty}^{\infty}\mathbb{P}\left(X_{\infty}\leq
\frac{x}{e^{\sigma\sqrt{\tau}z+(m-\frac{1}{2}\sigma^{2})\tau}}-1\right)
\frac{1}{\sqrt{2\pi}}e^{-\frac{z^{2}}{2}}dz.
\end{equation}
Differentiating w.r.t. $x$, we get
\begin{equation}
f(x;\beta,\rho)
=\int_{-\infty}^{\infty}f\left(\frac{x}{e^{\sigma\sqrt{\tau}z+(m-\frac{1}{2}\sigma^{2})\tau}}-1;\beta,\rho\right)
\frac{\frac{1}{\sqrt{2\pi}}e^{-\frac{z^{2}}{2}}dz}{e^{\sigma\sqrt{\tau}z+(m-\frac{1}{2}\sigma^{2})\tau}}.
\end{equation}
Change the variable and let $w=e^{\sigma\sqrt{\tau}z+(m-\frac{1}{2}\sigma^{2})\tau}$, 
we get
\begin{equation}
f(x;\beta,\rho)= \int_0^x \frac{dw}{\sqrt{2\pi\sigma^2\tau} w^2}
e^{-\frac{1}{2\sigma^2\tau}(\log w - (m-\frac12\sigma^2)\tau)^2} f(x/w-1;\beta,\rho). \nonumber
\end{equation}
Finally, make the change of variable $y=x/w-1$ in the last integral, and
use the definition $\beta=\sigma^{2}\tau$ and $\rho=m\tau$, we get (\ref{integraleq}).
\end{proof}

In the next section we study the asymptotic form of the density function
$f(x;\beta,\rho)$, in the limits of very small and very large argument $x$. 
In section~\ref{NumericalSection} we solve the integral equation 
(\ref{integraleq}) numerically, and investigate the dependence of the solution
$f(x;\beta,\rho)$
on the parameters $\beta,\rho$.

\subsection{Asymptotics of the Distribution of $X_\infty$}

We study here the tail asymptotics of the infinite sum of GBM. 

\subsubsection{Small $x$ asymptotics of $\mathbb{P}(X_{\infty}\leq x)$}

The limiting density function $f(x)$ falls off faster than 
any power of $x$ as $x\to 0$. We start by proving that all inverse moments
$\mathbb{E}[X_\infty^{-n}]$ with $n\in \mathbb{N}$ exist and are finite. 

\begin{proposition}\label{prop:3}
The inverse integer moments of $X_\infty$ of all orders
are finite and are bounded from above as
\begin{equation}
\mathbb{E}[X_\infty^{-n}] = e^{\frac12 \sigma^2 \tau n (n+1) - n m\tau}
\mathbb{E}\left[(1+X_\infty)^{-n}\right] \leq 
e^{\frac12 \sigma^2 \tau n (n+1) - n m\tau}\,,\quad
n \in \mathbb{N}\,.
\end{equation}
\end{proposition}

\begin{proof}
Follows by taking the inverse of the functional relation (\ref{Xrel}) raised 
to power $n$, and taking the expectation of both sides. The last step follows
from the positivity of $X_\infty$.
\end{proof}

>From Proposition~\ref{prop:3} we can prove that the probability that $X_\infty
\leq \epsilon$ is smaller than any power $\epsilon^n$ up to a constant. 
Using the Chebyshev inequality we find
\begin{equation}\label{upperEp}
\mathbb{P}(X_\infty \leq \epsilon) \leq \epsilon^n \mathbb{E}[X_\infty^{-n}]
\leq \epsilon^n e^{\frac12 \sigma^2 \tau n (n+1) - n m\tau}\,,\quad
n \in \mathbb{N}\,.
\end{equation}

\begin{remark}
These properties are similar to those of the infinite time-integral of the 
GBM. The inverse Gamma distribution (\ref{invGamma}) has finite negative
moments of all orders. This follows from the fact that $e^{-1/x}$ falls off
faster than any power of $x$ as $x\to 0$. 
\end{remark}

The result \eqref{upperEp} gives only an upper bound on the left tail of 
$X_{\infty}$. We can indeed precisely determine the leading order asymptotics 
for the left tail of $X_{\infty}$, which is given by the following result:

\begin{proposition}\label{LeftTailProp}
\begin{equation}
\lim_{\epsilon\rightarrow 0}\frac{\log\mathbb{P}(X_{\infty}\leq\epsilon)}{(\log\epsilon)^{2}}
=-\frac{1}{2\sigma^{2}\tau}\,.
\end{equation}
\end{proposition}

\begin{proof}
Recall that $X_{\infty}=\mathcal{A}(1+X_{\infty})$ in distribution and $X_{\infty}$ is independent of 
$\mathcal{A}:=e^{\sigma\sqrt{\tau}Z+(m-\frac{1}{2}\sigma^{2})\tau}$, 
where $Z\sim N(0,1)$. 
For any $\epsilon>0$, since $X_{\infty}\geq 0$ a.s.,
\begin{equation}
\mathbb{P}(X_{\infty}\leq\epsilon)=\mathbb{P}(\mathcal{A}(1+X_{\infty})\leq\epsilon)
\leq\mathbb{P}(\mathcal{A}\leq\epsilon).
\end{equation}
On the other hand, for any $\delta>0$, 
\begin{align}
\mathbb{P}(X_{\infty}\leq\epsilon)&=\mathbb{P}(\mathcal{A}(1+X_{\infty})\leq\epsilon)
\\
&\geq\mathbb{P}(\mathcal{A}(1+X_{\infty})\leq\epsilon|0\leq X_{\infty}\leq\delta)\mathbb{P}(X_{\infty}\leq\delta)
\nonumber
\\
&\geq\mathbb{P}\left(\mathcal{A}\leq\frac{\epsilon}{1+\delta}\bigg|0\leq X_{\infty}\leq\delta\right)\mathbb{P}(X_{\infty}\leq\delta)
\nonumber
\\
&=\mathbb{P}\left(\mathcal{A}\leq\frac{\epsilon}{1+\delta}\right)\mathbb{P}(X_{\infty}\leq\delta),
\nonumber
\end{align}
where the last step above used the independency of $X_{\infty}$ and $\mathcal{A}$. 

Therefore, for any $\epsilon,\delta>0$,
\begin{equation}\label{LowerUpper}
\log\mathbb{P}\left(\mathcal{A}\leq\frac{\epsilon}{1+\delta}\right)
  +\log\mathbb{P}(X_{\infty}\leq\delta)
\leq\log\mathbb{P}(X_{\infty}\leq\epsilon)
\leq\log\mathbb{P}(\mathcal{A}\leq\epsilon).
\end{equation}
Since $\mathcal{A}$ is log-normally distributed, we have
\begin{equation}
\mathbb{P}(\mathcal{A}\leq\epsilon)=
\mathbb{P}\left(\sigma\sqrt{\tau}Z\leq\log\epsilon-\left(m-\frac{1}{2}\sigma^{2}\right)\tau\right)
=\Phi\left(\frac{1}{\sigma\sqrt{\tau}} \left(\log\epsilon -
\left(m-\frac{1}{2}\sigma^{2}\right)\tau\right) \right),
\end{equation}
where $\Phi(x) = \int_{-\infty}^x \frac{dt}{\sqrt{2\pi}} e^{-\frac12 t^2}$ is
the well-known normal cumulative distribution function.
Using the asymptotic expansion for $\Phi(x)$ with $x\to -\infty$ we get
\begin{equation}\label{Atail}
\lim_{\epsilon\rightarrow 0}\frac{\log\mathbb{P}
  (\mathcal{A}\leq\epsilon)}{(\log\epsilon)^{2}}
=-\frac{1}{2\sigma^{2}\tau}\,.
\end{equation}
Taking $\epsilon\to 0$ at fixed $\delta$ in \eqref{LowerUpper}, and using
(\ref{Atail}), we conclude that
\begin{equation}
\lim_{\epsilon\rightarrow 0}\frac{\log\mathbb{P}(X_{\infty}\leq\epsilon)}{(\log\epsilon)^{2}}
=-\frac{1}{2\sigma^{2}\tau}\,.
\end{equation}
\end{proof}

\begin{remark}
In the continuous time case, $Y_{\infty}$ is inverse Gamma distributed and
\begin{equation}
\lim_{\epsilon\rightarrow 0}\epsilon\log\mathbb{P}(Y_{\infty}\leq\epsilon)
=-\frac{2}{\sigma^{2}} \,.
\end{equation}
We see that the left tail asymptotics of $X_{\infty}$ is different 
from that of (inverse Gamma distributed) $Y_{\infty}$. The density of the
discrete time sum $X_\infty$ is less suppressed near the $x\to 0$ point than
the density of the continuous time integral $Y_{\infty}$.
\end{remark}

\subsubsection{Large $x$ asymptotics of $\mathbb{P}(X_{\infty}\geq x)$}

The tail behavior of the limit distribution of the recursion (\ref{Kesten})
has been well studied in the literature. We summarize here the main results,
see \cite{MST} for a recent review of applications. Define the function
\begin{equation}
\varphi(\kappa) = \mathbb{E}[(a_1)^\kappa]\,,\qquad \kappa \in \mathbb{R}\,.
\end{equation}
If there exists a positive $\alpha$ such that $\varphi(\alpha)=1$, 
$\mathbb{E}[(a_1)^\alpha \log a_1], \mathbb{E}[|b_1|^\alpha]$ are both
finite, the law of $\log a_1$ is non-arithmetic\footnote{An arithmetic
distribution is defined as a distribution which has support only on
integer multiples of some real number.} and for every $x$, 
$\mathbb{P}(a_1 x + b_1 = x ) < 1$, then there exists a constant $c>0$
and
\begin{equation}
\mathbb{P}(x_\infty > x) \sim c x^{-\alpha} \,.
\end{equation}

For our case we have
\begin{equation}
\varphi(\kappa) = \mathbb{E}[(\mathcal{A})^\kappa] =
\mathbb{E}\left[e^{\kappa \sigma \sqrt{\tau} z + \kappa(m-\frac12
\sigma^2)\tau}\right] = e^{\frac12 \beta \kappa(\kappa-1) + \kappa \rho}\,.
\end{equation}
The equation $\varphi(\alpha)=1$ has solutions $\alpha=0$ and
\begin{equation}
\alpha = 1 - \frac{2m}{\sigma^2} =1 - \frac{2\rho}{\beta}  \,.
\end{equation}
The remaining technical conditions being satisfied, we find that
the limiting distribution of $X_\infty$ has the following tail behavior:

\begin{proposition}\label{PropXinftyRightTail}
The right tail asymptotic behavior of $X_\infty$ is given by the following
relation
\begin{equation}\label{tail}
\mathbb{P}(X_\infty > x) \sim c x^{-1 + \frac{2\rho}{\beta}}\,.
\end{equation}
\end{proposition}

\begin{remark}
The inequality (\ref{ineq}) gives that $\alpha > 0$, such that the cumulative
distribution of $X_\infty$ is guaranteed to decrease to zero as $x\to \infty$.
This ensures that the distribution of $X_\infty$ is normalizable to 1.
\end{remark}

\begin{remark}
The tail asymptotics (\ref{tail}) is identical with the tail asymptotics
of the continuous time integral of the GBM for $T\to \infty$ obtained in 
Theorem~\ref{Thm:Dufresne},
as can be seen from the explicit result for the density of $X_\infty$ given
in (\ref{invGamma}).
\end{remark}

The constant $c$ in \eqref{tail} can be also estimated. For a general 
linear recursion
$x_{i+1}=a_i x_i + b_i$ with i.i.d. $(a_i,b_i)$, we have the following results.
If $a_1,b_1\geq 0$, the constant $c$ is given by \cite{Goldie}
\begin{equation}\label{cdef}
c = \frac{\mathbb{E}[(a_1 x_0+b_1)^\alpha]-\mathbb{E}[(a_1 x_0)^\alpha]}
{\alpha \mathbb{E}[a_1^\alpha \log a_1]}\,,
\end{equation}
where $x_0$ is distributed according to the stationary law of $x_n$, and is 
independent of $(a_1,b_1)$. 

For our case we have  $a_1 = b_1 = \mathcal{A} =
e^{\sigma\sqrt{\tau} Z + (m-\frac12 \sigma^2)\tau}$. The coefficient becomes
\begin{equation}\label{crel}
c = \frac{\mathbb{E}[\mathcal{A}^\alpha [(1+X_\infty)^\alpha-X_\infty^\alpha]]}
{\alpha \mathbb{E}[\mathcal{A}^\alpha \log \mathcal{A}]}.
\end{equation}
Using $\alpha = 1 - \frac{2m}{\sigma^2}$ we have
\begin{eqnarray}
\mathbb{E}[\mathcal{A}^\alpha \log \mathcal{A}] &=& 
\mathbb{E}\left[\left(\sigma\sqrt{\tau} Z + \left(m-\frac12 \sigma^2\right)\tau\right) 
e^{\alpha \sigma\sqrt{\tau} Z + \alpha (m-\frac12\sigma^2) \tau}\right] \\
&=& \frac12 \sigma^2 \tau - m\tau   > 0 \,. \nonumber
\end{eqnarray}
This is always positive under the constraint $m < \frac12 \sigma^2$.

The numerator in (\ref{crel}) factors into two expectations, by independence
of $\mathcal{A}, X_\infty$. The expectation over $X_\infty$ is not easy 
to compute in the general case. It can be expressed as an integral over the 
density of $X_\infty$
\begin{equation}\label{expX}
\mathbb{E}[(1+X_\infty)^\alpha-X_\infty^\alpha] = 
\int_0^\infty dx f(x;\beta,\rho) [(1+x)^\alpha-x^\alpha],
\end{equation}
which in general has to be evaluated numerically, using the solution for 
$f(x;\beta,\rho)$ obtained by solving the integral equation.

The case of positive integer $\alpha\in \mathbb{N}$ is simpler, as the 
expectation (\ref{expX}) is a linear combination of the first $\alpha-1$
positive integer moments of $X_\infty$.
In Section~\ref{Sec:PositiveMoments} it is shown that these moments can be 
evaluated in closed form (whenever they exist). 
The case $m=0$ is particularly simple, as we have $\alpha=1$. 
For this case the expectation (\ref{expX}) is 1, and we get
\begin{equation}
c = \frac{2}{\sigma^2\tau}\,.
\end{equation}
The right tail asymptotics for the density of $X_\infty$ for $m=0$ is
\begin{equation}
f(x;\sigma,0) \sim -\frac{d}{dx} (c/x) = \frac{c}{x^2} = 
\frac{2}{\sigma^2 \tau x^2} \,.
\end{equation}
The corresponding tail asymptotics for $\tau X_\infty$ is
\begin{equation}
g(x;\sigma,0) = \frac{1}{\tau} p(x/\tau;\sigma,0) = \frac{2}{\sigma^2 x^2}\,.
\end{equation}
This is identical with the tail behavior of the exact continuous time 
integral of the GBM which is the given by the inverse Gamma density function
$\phi_\infty(x;\sigma,0)$ defined in (\ref{invGamma}).
However, for general $m$ this simple result does not hold, as can be
shown for example by taking $\alpha=2$ and using the results of 
Section~\ref{Sec:PositiveMoments}.

\subsection{Limiting Distribution for $\tau\rightarrow 0$}

It is natural to ask if the Riemann sum $\tau X_{\infty}$ converges
to $Y_{\infty}$ in distribution as $\tau\rightarrow 0$. We have the following
result.

\begin{theorem}\label{InftyThm}
(i) For any $m,\sigma$ real numbers, we have
\begin{equation}
\tau\sum_{i=1}^{N}e^{\sigma W_{t_{i-1}}+(m-\frac{1}{2}\sigma^{2})t_{i-1}}
\rightarrow\int_{0}^{T}e^{\sigma W_{t}+(m-\frac{1}{2}\sigma^{2})t}dt,
\end{equation}
in $L^{1}$ norm as $\tau\rightarrow 0$ at $N\tau = T$ fixed.

(ii) Furthermore, provided  that $m<0$, we have also 
\begin{equation}
\tau\sum_{i=1}^{\infty}e^{\sigma W_{t_{i-1}}+(m-\frac{1}{2}\sigma^{2})t_{i-1}}
\rightarrow\int_{0}^{\infty}e^{\sigma W_{t}+(m-\frac{1}{2}\sigma^{2})t}dt,
\end{equation}
in $L^{1}$ norm as $\tau\rightarrow 0$.
\end{theorem}

\begin{proof}
The proof is given in Appendix~\ref{App:proofs}.
\end{proof}

This result implies that $\tau X_{\infty}$ can be approximated by 
$\int_{0}^{\infty}e^{\sigma W_{t}+(m-\frac{1}{2}\sigma^{2})t}dt$ for sufficiently small values of $\tau$,
and we know from the literature that $\int_{0}^{\infty}e^{\sigma W_{t}+(m-\frac{1}{2}\sigma^{2})t}dt$
follows an inverse Gamma distribution. Indeed, our discrete approximation approach
gives an alternative proof that the limiting continuous integral indeed follows an inverse Gamma distribution.

\begin{proposition}
Assuming $m-\frac12\sigma^2<0$ we have the following limit in distribution:
\begin{equation}
\lim_{\tau \to 0} \frac{1}{\tau X_\infty(\sigma^2\tau,m\tau)} = 
\mathbf{Gamma}\left(1 - \frac{2m}{\sigma^2},\frac12\sigma^2\right) \,.
\end{equation}
\end{proposition}

\begin{proof}
Let $Y_{\tau}:=\tau X_{\infty}$. 
Note that since $Y_{\tau}=\tau X_{\infty}$, we have 
$Y_{\tau} = \mathcal{A} (\tau+Y_{\tau})$ 
in distribution, where $Y_{\tau}$ is independent of $\mathcal{A}$. 
The Laplace transform of $\log Y_{\tau}$ is given by
\begin{equation}\label{sub0}
\mathbb{E}[Y_{\tau}^{-\theta}]=
\mathbb{E}[\mathcal{A}^{-\theta}]\mathbb{E}[(\tau+Y_{\tau})^{-\theta}],
\qquad
\theta>0,
\end{equation}
where we used the independence of $Y_{\tau}$ and $\mathcal{A}$. 
Note that since 
$\mathcal{A}=
e^{\sigma\sqrt{\tau}Z+(m-\frac{1}{2}\sigma^{2})\tau}$ is log-normally distributed, we have
\begin{equation}\label{sub1}
\mathbb{E}[\mathcal{A}^{-\theta}]=e^{\frac{\theta^{2}+\theta}{2}\sigma^{2}\tau-\theta m\tau}
=1+\frac{\theta^{2}+\theta}{2}\sigma^{2}\tau-\theta m\tau+O(\tau^{2}),
\end{equation}
and moreover, 
\begin{equation}\label{sub2}
\mathbb{E}[(\tau+Y_{\tau})^{-\theta}]
-\mathbb{E}[Y_{\tau}^{-\theta}]
=-\theta\tau \mathbb{E}[(Y_{\tau})^{-\theta-1}]+O(\tau^{2}).
\end{equation}
Substituting \eqref{sub1} and \eqref{sub2} into \eqref{sub0}, we get
\begin{equation}
\mathbb{E}[Y_{\tau}^{-\theta}]
=\left(1+\frac{\theta^{2}+\theta}{2}\sigma^{2}\tau-\theta m\tau+O(\tau^{2})\right)
\left(\mathbb{E}[Y_{\tau}^{-\theta}]
-\theta\tau \mathbb{E}[(Y_{\tau})^{-\theta-1}]+O(\tau^{2})\right).
\end{equation}
Let $\tau\rightarrow 0$, and define $Y_0 = \lim_{\tau \to 0} Y_\tau$. 
In this limit the coefficient of $O(\tau)$ term must vanish and therefore 
we get the identity
\begin{equation}\label{5}
\mathbb{E}[Y_{0}^{-\theta-1}]=
\left(\frac{\theta+1}{2}\sigma^{2}-m\right)\mathbb{E}[Y_{0}^{-\theta}].
\end{equation}
Note that the above $\tau\rightarrow 0$ limit is not only valid for $\theta>0$
but also for $\theta<0$ and $|\theta|$ sufficiently small. 
Therefore, we proved the convergence of the m.g.f. of $-\log Y_{\tau}$,
that is the convergence of $\mathbb{E}[Y_{\tau}^{-\theta}]$ for $\theta$ in a neighborhood
of $0$, which implies the convergence of $Y_{\tau}$ to $Y_{0}$ in distribution.

Let us verify that the inverse Gamma distribution indeed satisfies this
identity. Using the density function in (\ref{invGamma}) we can compute the
expectations appearing in (\ref{5}) in closed form
\begin{align}
\mathbb{E}[Y_{0}^{-\theta}]
&=\int_{0}^{\infty}y^{-\theta}\left(\frac{2}{\sigma^{2}}\right)^{1-\frac{2m}{\sigma^{2}}}\left(\frac{1}{y^{2}}\right)^{1-\frac{m}{\sigma^{2}}}
\frac{1}{\Gamma(1-\frac{2m}{\sigma^{2}})}
e^{-\frac{2}{\sigma^{2}y}}dy
\\
& = \frac{\Gamma(\theta+\beta)}{\Gamma(\beta)} \left( \frac{\sigma^2}{2}\right)^\theta
=\frac{1}{\frac{\theta+1}{2}\sigma^{2}-m}\mathbb{E}[Y_{0}^{-\theta-1}].
\nonumber
\end{align}
with $\beta = 1 - \frac{2m}{\sigma^2}$. This reproduces indeed the relation (\ref{5}).

Next, we show that \eqref{5} indeed implies that $Y_{0}$ is inverse Gamma 
distributed. Denote $W_\tau = \frac{1}{Y_\tau}$. We will show that 
$W_0 = \lim_{\tau \to 0}W_\tau$ is distributed as a Gamma distribution
\begin{equation}\label{result}
W_0 = \mathbf{Gamma}\left(1 - \frac{2m}{\sigma^2}, \frac12\sigma^2\right)\,.
\end{equation}

Define the moment generating function of $W_0$ as
\begin{equation}\label{mgf}
M(t) = \mathbb{E}[e^{t W_0}] \,.
\end{equation}
Denote $\theta = j$ in \eqref{5}, multiply both sides of this equation
with $\frac{t^j}{j!}$, and sum over $j \in \mathbb{N}$. The sums can be 
expressed in terms of the m.g.f. of $W_0$ as
\begin{align}
& \sum_{j=0}^\infty \frac{t^j}{j!} \mathbb{E}[W_0^j] = M(t), \\
& \sum_{j=0}^\infty \frac{j t^j}{j!} \mathbb{E}[W_0^j] = t \frac{d}{dt} M(t).
\end{align}

The relation \eqref{5} becomes a differential equation for the function
$M(t)$
\begin{equation}
M'(t) = \left(\frac12 \sigma^2 - m\right) M(t) + \frac12 \sigma^2 t M'(t),
\qquad M(0)=1.
\end{equation}
This is a first-order linear ODE, which yields the solution
\begin{equation}
M(t) = \left(1 - \frac12 \sigma^2 t \right)^{-1+\frac{2m}{\sigma^2}},
\qquad t<\frac{2}{\sigma^{2}}.
\end{equation}
This has precisely the same form as the m.g.f. of a Gamma distributed random 
variable with the parameters shown in (\ref{result}). This proves that 
$W_0 = 1/Y_0$ is distributed as a Gamma random variable, 
and thus $Y_0$ follows an inverse Gamma distribution.
\end{proof}

This result implies that the cumulative distribution function of $X_\infty$,
defined as
\begin{equation}
F(x;\beta,\rho) = \mathbb{P}(X_\infty < x) = \int_0^x f(y; \beta,\rho) dy
\end{equation}
has the following limiting behavior.

\begin{proposition}\label{prop:smalltau}
We have
\begin{equation}\label{limf}
\lim_{\tau \to 0} F(x/\tau;\sigma^2 \tau,m \tau) = 
\Phi_\infty(x;\sigma,m),
\end{equation}
with $\Phi_\infty(x;\sigma,m)$ given in (\ref{invGamma2}).
This can be expressed alternatively as a limiting result for the 
cumulative distribution function
$F(x;\beta,\rho)$ as the parameters $\beta, \rho \to 0$ at fixed ratio
$\rho/\beta$. We have
\begin{equation}\label{limf2}
\lim_{\substack{\beta,\rho\to 0,
\\
\rho/\beta = \mbox{fixed}}}
F(x;\beta,\rho) = \int_0^x \frac{dy}{y} 
    \left( \frac{2}{\beta y}\right)^{1-\frac{2\rho}{\beta}}
\frac{1}{\Gamma(1-\frac{2\rho}{\beta})} e^{-\frac{2}{\beta y}}
= \frac{1}{\Gamma(1-\frac{2\rho}{\beta})}
\Gamma\left( 1 - \frac{2\rho}{\beta}; \frac{2}{\beta x} \right)\,.
\end{equation}
\end{proposition}


\section{Geometric Mortality}\label{GeometricSection}

In the previous section we have studied the infinite sum of geometric 
Brownian motions  $X_{\infty}=\sum_{i=1}^{\infty}e^{\sigma W_{t_{i}}+
(m-\frac{1}{2}\sigma^{2})t_{i}}$, where $t_{i}=i\tau$. 
A more realistic modeling of a stochastic annuity takes into account that
the total time over which the annuity is paid is random. 
In the discrete time setting, one of the simplest assumptions
is that the number of periods that a person will live follows a geometric 
distribution, that is, given that the person is still alive at present, the 
probability that he/she is still alive
at the next time step is $1-p$, with $0<p<1$. In other words,
we are interested in the distributions of 
\begin{equation}
X_{N}=\sum_{i=1}^{N}e^{\sigma W_{t_{i}}+(m-\frac{1}{2}\sigma^{2})t_{i}},
\end{equation}
where $N$ follows a geometric distribution, independent of the Brownian 
motion $W_{t}$, that is,
\begin{equation}
\mathbb{P}(N=k)=(1-p)^{k-1}p,
\qquad
k=1,2,3,\ldots.
\end{equation}
Therefore, it is not hard to see that
\begin{equation}\label{XNlaw}
X_{N}=\mathcal{AQ} + \mathcal{A}(1-\mathcal{Q})(1+X_{N}),
\end{equation}
in distribution, where $X_{N}$, $\mathcal{A}$ and $\mathcal{Q}$ are independent, 
with $\mathcal{A} :=e^{\sigma\sqrt{\tau}Z + (m-\frac{1}{2}\sigma^{2})\tau}$, 
where $Z\sim N(0,1)$ and $\mathcal{Q} = \{0, 1 \}$ is a Bernoulli random variable
taking values $0,1$ with probabilities
\begin{equation}
\mathbb{P}(\mathcal{Q}=1)=p=1-\mathbb{P}(\mathcal{Q}=0).
\end{equation}

\subsection{Probability Density Function and Tail Asymptotics for $X_{N}$}

We start by deriving an integral equation for the probability density function 
$f(x;\beta,\rho,p)$ for $X_{N}$. For any $x>0$, 
\begin{align}
\mathbb{P}(X_{N}\leq x)&=\mathbb{P}(\mathcal{AQ+A}(1-\mathcal{Q})(1+X_{N})\leq x)
\\
&=
p\mathbb{P}(\mathcal{A} \leq x)
+(1-p)\mathbb{P}(\mathcal{A}(1+X_{N})\leq x)
\nonumber
\\
&=p\mathbb{P}\left(Z\leq\frac{\log x-m\tau+\frac{1}{2}\sigma^{2}\tau}{\sigma\sqrt{\tau}}\right)
\nonumber
\\
&\qquad\qquad
+(1-p)\int_{-\infty}^{\infty}\mathbb{P}\left(X_{N}\leq
\frac{x}{e^{\sigma\sqrt{\tau}z+(m-\frac{1}{2}\sigma^{2})\tau}}-1\right)
\frac{1}{\sqrt{2\pi}}e^{-\frac{z^{2}}{2}}dz.
\nonumber
\end{align}
Differentiating w.r.t. $x$, we get
\begin{align}
f(x;\beta,\rho,p)
&=p\frac{1}{\sqrt{2\pi}x}
e^{-\frac{1}{2\sigma^{2}\tau}(\log x-m\tau+\frac{1}{2}\sigma^{2}\tau)^{2}}
\\
&\qquad\qquad
+(1-p)\int_{-\infty}^{\infty}f\left(\frac{x}{e^{\sigma\sqrt{\tau}z+(m-\frac{1}{2}\sigma^{2})\tau}}-1;\beta,\rho,p\right)
\frac{\frac{1}{\sqrt{2\pi}}e^{-\frac{z^{2}}{2}}dz}{e^{\sigma\sqrt{\tau}z+(m-\frac{1}{2}\sigma^{2})\tau}},
\nonumber
\end{align}
where we recall that
\begin{equation*}
\beta=\sigma^{2}\tau,
\qquad
\rho=m\tau.
\end{equation*}
Change the variable and let $w=e^{\sigma\sqrt{\tau}z+(m-\frac{1}{2}\sigma^{2})\tau}$, 
we get
\begin{align}
f(x;\beta,\rho,p)&= 
p\frac{1}{\sqrt{2\pi}x}
e^{-\frac{1}{2\sigma^{2}\tau}(\log x-m\tau+\frac{1}{2}\sigma^{2}\tau)^{2}}
\\
&\qquad\qquad
+(1-p)\int_0^x \frac{dw}{\sqrt{2\pi\sigma^2\tau} w^2}
e^{-\frac{1}{2\sigma^2\tau}(\log w - (m-\frac12\sigma^2)\tau)^2} f(x/w-1;\beta,\rho,p). \nonumber
\end{align}
Finally, make the change of variable $y=x/w-1$ in the last integral, 
and use the definitions $\beta=\sigma^{2}\tau$ and $\rho=m\tau$, we get:

\begin{proposition}\label{prop:14}
The density function of $X_N$, where $N$ follows a geometric distribution
with parameter $p$, satisfies the integral equation
\begin{align}\label{IntegralEqXN}
f(x;\beta,\rho,p) &= 
p\frac{1}{\sqrt{2\pi\beta}x}
e^{-\frac{1}{2\beta}(\log x-\rho+\frac{1}{2}\beta)^{2}}
\\
&\qquad
+(1-p)\frac{1}{x} \int_0^\infty \frac{dy}{\sqrt{2\pi\beta}} 
\exp\left( -\frac{1}{2\beta} 
\left[ \log \frac{x}{1+y} + \frac12 \beta - \rho \right]^2
\right)
f(y;\beta,\rho,p)\,,
\nonumber
\end{align}
where $\beta=\sigma^{2}\tau$ and $\rho=m\tau$.
\end{proposition}

Next, let us derive the right and left tails of $X_{N}$.
Define the function 
\begin{equation}
\varphi(\kappa)=\mathbb{E}\left[\left( \mathcal{A}(1-\mathcal{Q}) \right)^{\kappa}\right]
=\mathbb{E}[\mathcal{A}^{\kappa}]\mathbb{E}[(1-\mathcal{Q})^{\kappa}].
\end{equation}
We can compute that
\begin{equation}
\varphi(\kappa)=
e^{\frac{1}{2}\beta \kappa(\kappa-1)+\kappa \rho}(1-p).
\end{equation}
The equation $\varphi(\mu)=1$ has the positive solution
\begin{equation}\label{mugeom}
\mu=\frac{- \rho + \frac{1}{2}\beta +
\sqrt{(\rho - \frac{1}{2}\beta)^{2} - 2\beta\log(1-p)}}{\beta},
\end{equation}
and therefore, as $x\rightarrow+\infty$,
\begin{equation}\label{XNrighttail}
\mathbb{P}(X_{N}>x)\sim c_{+}x^{-\mu},
\end{equation}
for some constant $c_{+}>0$. 
This constant is given by (\ref{cdef}) \cite{Goldie},
which can be expressed as
\begin{equation}\label{cplus}
c_+ = \frac{\mathbb{E}[\mathcal{A}^\mu]}
           {\mu(1-p)\mathbb{E}[\mathcal{A}^\mu \log A]}
\left\{ p + (1-p) \left[ \mathbb{E}[(X_N+1)^\mu] - \mathbb{E}[(X_N)^\mu]
\right] \right\} \,.
\end{equation}
The first factor can be evaluated further as
\begin{eqnarray}
&& \frac{\mathbb{E}[\mathcal{A}^\mu]}{\mu(1-p)\mathbb{E}[\mathcal{A}^\mu \log A]} = 
\frac{1}{\mu(1-p)^2}\frac{1}{\mathbb{E}[\mathcal{A}^\mu \log \mathcal{A}]} \\
&& \quad = \frac{1}{\mu(1-p)(\rho - \frac12\beta + \mu \beta)}
= \frac{1}{\mu(1-p) \sqrt{(\rho - \frac{1}{2}\beta)^{2} - 2\beta\log(1-p)} }
\,.\nonumber
\end{eqnarray}

\begin{remark}
The technical conditions (\ref{cond12}) are satisfied, with 
$\mathbb{E}[\log\{\mathcal{A}(1-\mathcal{Q})\}]=-\infty$. We note that this relaxes the
condition on the drift $m < \frac12\sigma^2$ which is present for $p=0$.
For $p \neq 0$, the solution for $\mu$ in (\ref{mugeom}) 
is always strictly positive,  for all $m\in \mathbb{R}$. 
For $\rho < \frac12 \beta$ we have the stronger lower bound 
$\mu > 1 - \frac{2\rho}{\beta}$.
\end{remark}

Let us also derive the left tail asymptotics for $X_{N}$. 
Note that $X_{N}$ is non-negative. Therefore, for any $\epsilon>0$,
\begin{equation}
\mathbb{P}(X_{N}\leq\epsilon)
=\mathbb{P}(\mathcal{AQ}+\mathcal{A}(1-\mathcal{Q})(1+X_{N})\leq\epsilon)
\leq\mathbb{P}(\mathcal{A}\leq\epsilon).
\end{equation}
On the other hand, for any $\delta>0$,
\begin{align}
\mathbb{P}(X_{N}\leq\epsilon)
&\geq
\mathbb{P}(\mathcal{AQ+A}(1-\mathcal{Q})(1+X_{N})\leq
 \epsilon|\mathcal{Q}=0,0\leq X_{N}\leq\delta)
\mathbb{P}(\mathcal{Q}=0,X_{N}\leq\delta)
\\
&\geq(1-p)\mathbb{P}(X_{N}\leq\delta)
\mathbb{P}\left(\mathcal{A}\leq\frac{\epsilon}{1+\delta}\right).
\nonumber
\end{align}
Following the proofs for the left tail asymptotics for $X_{\infty}$, we 
conclude that
\begin{equation}
\lim_{\epsilon\rightarrow 0}\frac{\log\mathbb{P}(X_{N}\leq\epsilon)}{(\log\epsilon)^{2}}
=-\frac{1}{2\sigma^{2}\tau}\,.
\end{equation}

\begin{remark}
As $p\rightarrow 0$, the person will live forever, and the time horizon
of the stochastic annuity becomes infinite. Therefore, 
as $p\rightarrow 0$, we expect that $X_{N}\rightarrow X_{\infty}$ in distribution. 
\end{remark}

\subsection{The limit $\tau \to 0$ and comparison with the continuous time}
Let $Y_{\tau}:=\tau X_{N}$ and $p=\lambda\tau$. 
As $\tau\rightarrow 0$, we expect that $Y_{\tau}$ 
converges to $\int_{0}^{T_\lambda}e^{\sigma W_{t}-\frac{1}{2}\sigma^{2}t+mt}dt$ 
in distribution, where $T_\lambda$ is exponentially distributed with parameter 
$\lambda>0$ and is independent of the Brownian motion $W_{t}$.

Indeed, we will prove the following result. 

\begin{theorem}\label{ThmExponentialT}
Let $T_{\lambda}$ be an exponentially distributed random variable with 
parameter $\lambda>0$, $N$ be a geometric distributed random variable with
parameter $p=\lambda \tau$, both assumed to be independent of the Brownian 
motion $W_{t}$. Then, assuming $m< \lambda$, we have
\begin{equation}
\tau\sum_{i=1}^{N}e^{\sigma W_{t_{i-1}}+(m-\frac{1}{2}\sigma^{2})t_{i-1}}
\rightarrow\int_{0}^{T_{\lambda}}e^{\sigma W_{t}+(m-\frac{1}{2}\sigma^{2})t}dt,
\end{equation}
in distribution as $\tau\rightarrow 0$.
\end{theorem}

\begin{proof}
See Appendix~\ref{App:proofs}.
\end{proof}

The distribution of the time integral of the GBM up to an exponentially distributed
time is known in closed form \cite{Yor0}, see Appendix A
for a summary of the results. We will show here briefly how these well-known
results are consistent with the identity in law (\ref{XNlaw}) in the $\tau \to 0$ limit.

Note that since $Y_{\tau}=\tau X_{N}$, the relation (\ref{XNlaw}) gives
\begin{equation}
Y_{\tau}=\mathcal{A}(\tau+(1-\mathcal{Q})Y_{\tau})
\end{equation}
in distribution, where $Y_{\tau}$ is independent of $\mathcal{A,Q}$. 
The m.g.f. of $\log Y_{\tau}$ is given by (for $0<\theta<1$)
\begin{equation}\label{Qsub0}
\mathbb{E}[Y_{\tau}^{\theta}]=\mathbb{E}[\mathcal{A}^{\theta}]
  \mathbb{E}[(\tau+(1-\mathcal{Q})Y_{\tau})^{\theta}],
\qquad
0<\theta<1,
\end{equation}
where we used independence of $Y_{\tau}$ and $\mathcal{A}$. 
Since $\mathcal{A}=e^{\sigma\sqrt{\tau}Z+(m-\frac{1}{2}\sigma^{2})\tau}$ 
is log-normally distributed, we have
\begin{equation}\label{Qsub1}
\mathbb{E}[\mathcal{A}^{\theta}]=
e^{\frac{\theta^{2}-\theta}{2}\sigma^{2}\tau+\theta m\tau}
=1+\frac{\theta^{2}-\theta}{2}\sigma^{2}\tau+\theta m\tau+O(\tau^{2}),
\end{equation}
and moreover, 
\begin{align}\label{Qsub21}
\mathbb{E}[(\tau+(1-\mathcal{Q})Y_{\tau})^{\theta}]
&=\tau^{\theta}p+\mathbb{E}[(\tau+Y_{\tau})^{\theta}](1-p)
\\
&=\tau^{\theta}\lambda\tau+\mathbb{E}[(\tau+Y_{\tau})^{\theta}](1-\lambda\tau)
\nonumber
\end{align}
and
\begin{equation}\label{Qsub2}
\mathbb{E}[(\tau+Y_{\tau})^{\theta}]
-\mathbb{E}[(Y_{\tau})^{\theta}]
=\theta\mathbb{E}[(Y_{\tau})^{\theta-1}]\tau+O(\tau^{2}).
\end{equation}
Therefore, \eqref{Qsub21} and \eqref{Qsub2} imply that
\begin{equation}\label{Qsub3}
\mathbb{E}[(\tau+(1-\mathcal{Q})Y_{\tau})^{\theta}]
=\mathbb{E}[Y_{\tau}^{\theta}]+(\theta\mathbb{E}[Y_{\tau}^{\theta-1}]-\lambda\mathbb{E}[Y_{\tau}^{\theta}])\tau
+O(\tau^{1+\theta}).
\end{equation}
Substituting \eqref{Qsub1} and \eqref{Qsub3} into \eqref{Qsub0}, we get
\begin{equation}
\mathbb{E}[Y_{\tau}^{\theta}]
=\left(1+\frac{\theta^{2}-\theta}{2}\sigma^{2}\tau+\theta m\tau+O(\tau^{2})\right)
\left(\mathbb{E}[Y_{\tau}^{\theta}]+(\theta\mathbb{E}[Y_{\tau}^{\theta-1}]-\lambda\mathbb{E}[Y_{\tau}^{\theta}])\tau
+O(\tau^{1+\theta})\right).
\end{equation}
Let $\tau\rightarrow 0$, the coefficient of $O(\tau)$ term must vanish and therefore
the limit $Y_{0}$ must satisfy the identity
\begin{equation}\label{PreToCheck}
\left(\frac{\theta^{2}-\theta}{2}\sigma^{2}+\theta m-\lambda\right)\mathbb{E}[Y_{0}^{\theta}]
+\theta\mathbb{E}[Y_{0}^{\theta-1}]=0.
\end{equation}

Let
\begin{equation}
A_{t}^{(\mu)} := \int_{0}^{t}e^{2\mu s+2W_{s}}ds,
\qquad t\geq 0,\mu\in\mathbb{R}\,.
\end{equation}
It is well known that the law of this time integral of the GBM up to an
exponentially distributed random time $T_\lambda \sim \mathbf{Exp(\lambda)}$
is given by \cite{Yor0}
\begin{equation}\label{YorFormula}
2A_{T_{\lambda}}^{(\mu)}=\frac{B_{1,\alpha}}{G_{\beta}},
\end{equation}
in distribution, where $B_{1,\alpha}\sim\text{Beta}(1,\alpha)$
and $G_{\beta}\sim\Gamma(\beta,1)$ are independent random variables,
with parameters
\begin{equation}
\alpha=\frac{\mu}{2}+\frac{1}{2}\sqrt{2\lambda+\mu^{2}},
\qquad
\beta=-\frac{\mu}{2}+\frac{1}{2}\sqrt{2\lambda+\mu^{2}}.
\end{equation}

Taking $\sigma=2$ and $m=2\mu+2$ in (\ref{PreToCheck}), this relation becomes
\begin{equation}\label{ToCheck}
\left(2\theta^{2}+2\theta\mu-\lambda\right)\mathbb{E}[Y_{0}^{\theta}]
+\theta\mathbb{E}[Y_{0}^{\theta-1}]=0.
\end{equation}

We will prove next that this moment relation is satisfied indeed by
$Y_0=\frac{B_{1,\alpha}}{2G_\beta}$ as given by (\ref{YorFormula}).
Using the densities
of these random variables from (\ref{BetaPDF}), (\ref{GammaPDF}), 
we can compute that
\begin{equation}
\mathbb{E}[Y_{0}^{\theta-1}]
=\frac{\alpha}{2^{\theta-1}}
\int_{0}^{1}x^{\theta-1}(1-x)^{\alpha-1}dx
\int_{0}^{\infty}\frac{z^{-\beta-1}}{\Gamma(\beta)}z^{\theta-1}e^{-\frac{1}{z}}dz.
\end{equation}
The two integrals are evaluated as
\begin{equation}
\int_{0}^{1}x^{\theta-1}(1-x)^{\alpha-1}dx
= B(\theta,\alpha) \,,
\end{equation}
and
\begin{equation}
\int_{0}^{\infty}\frac{z^{-\beta-1}}{\Gamma(\beta)}
z^{\theta-1}e^{-\frac{1}{z}}dz
=\frac{\Gamma(\beta-\theta+1)}{\Gamma(\beta)} \, .
\end{equation}
Hence,
\begin{equation}
\mathbb{E}[Y_{0}^{\theta-1}]
=\frac{\alpha}{2^{\theta-1}}
B(\theta,\alpha) \frac{\Gamma(\beta-\theta+1)}{\Gamma(\beta)},
\end{equation}
and therefore
\begin{equation}
\mathbb{E}[Y_{0}^{\theta}]
=\frac{\alpha}{2^{\theta}} B(\theta+1,\alpha)
\frac{\Gamma(\beta-\theta)}{\Gamma(\beta)} \,.
\end{equation}

To check \eqref{ToCheck}, we need to show that
\begin{equation}
\left(2\theta^{2}+2\theta\mu-\lambda\right)B(\theta+1,\alpha)\Gamma(\beta-\theta)
+2\theta B(\theta,\alpha)\Gamma(\beta-\theta+1)=0.
\end{equation}
Since $B(\gamma,\delta)=\frac{\Gamma(\gamma)\Gamma(\delta)}{\Gamma(\gamma+\delta)}$, 
it is equivalent to show that
\begin{equation}
\left(2\theta^{2}+2\theta\mu-\lambda\right)\frac{\Gamma(\theta+1)}{\Gamma(\theta+1+\alpha)}\Gamma(\beta-\theta)
+2\theta\frac{\Gamma(\theta)}{\Gamma(\theta+\alpha)}\Gamma(\beta-\theta+1)=0,
\end{equation}
which is equivalent to show that
\begin{equation}
\left(\theta^{2}+\theta\mu-\frac{\lambda}{2}\right)\theta
+\theta(\beta-\theta)(\theta+\alpha)=0,
\end{equation}
which holds by the definition of $\alpha$ and $\beta$.
We conclude that the moment relation \eqref{ToCheck} is indeed satisfied by
the random variable
$Y_{0}=\frac{1}{2}\frac{B_{1,\alpha}}{G_{\beta}}$.

>From Theorem~\ref{ThmExponentialT} we have the following result.

\begin{proposition}\label{prop:XNsmalltau}
The cumulative distribution function of $X_N$ approaches the following 
limiting distribution as $\tau \to 0$
\begin{equation}
\lim_{\tau\to 0} 
\int_x^\infty
\frac{dz}{\tau} f\left(\frac{z}{\tau};\sigma^2\tau,m\tau,\lambda\tau\right) =
\int_x^\infty dz \phi_\lambda(z;\sigma,m,\lambda),
\end{equation}
with $\phi_\lambda(x;\sigma,m,\lambda)$ given in (\ref{Phiz}).
\end{proposition}

\begin{remark}
Convergence in distribution only implies the convergence
of the cumulative distribution function, not the probability density function.
Nevertheless, in practice, as we can see from Figure \ref{Fig:4}, the probability
density function for the discrete time case can be heuristically approximated
by the continuous time case via:
\begin{equation}\label{fNsmalltau}
f(x;\beta,\rho,p)\sim
\frac12 \beta \varphi\left(\frac12\beta x;a,b\right),
\qquad\text{as $\beta,\rho,p\to 0$},
\end{equation}
where $\varphi(x;a,b)$ is given in (\ref{varphiy}) and
\begin{eqnarray}
a &=& \frac{1}{2\beta}
\left( 2\rho - \beta + \sqrt{(2\rho-\beta)^2+8p \beta} \right), \\
b &=& \frac{1}{2\beta}
\left( - 2\rho + \beta + \sqrt{(2\rho-\beta)^2+8p\beta} \right).
\end{eqnarray}
\end{remark}

\section{Exponential L\'{e}vy Model}\label{LevySection}

We can replace the geometric Brownian motion by the exponential L\'{e}vy model,
and most of the results and conclusions of the previous section still hold. 
Let us consider the sum
$X_{N}=\sum_{i=1}^{N}e^{Z_{t_{i}}+m t_{i}}$, 
where $m<0$, $t_{i}=i\tau$ and $Z_{t_{i}}$ is a L\'{e}vy process
so that for any $\theta\in\mathbb{R}$, $\mathbb{E}[e^{\theta Z_{t}}]=e^{\kappa(\theta)t}$, 
and $N$ follows a geometric distribution with parameter $p$, 
independent of the L\'{e}vy process $Z_{t}$, that is,
\begin{equation}
\mathbb{P}(N=k)=(1-p)^{k-1}p,
\qquad
k=1,2,3,\ldots.
\end{equation}
Therefore, it is not hard to see that
\begin{equation}
X_{N}  =\mathcal{AQ} + \mathcal{A}(1-\mathcal{Q})(1+X_{N}),
\end{equation}
in distribution, where $X_{N}$, $\mathcal{A}$ and $\mathcal{Q}$ 
are independent, with\footnote{Note the change of definition for $\mathcal{A}$
compared to the previous sections.}
$\mathcal{A} := e^{Z_{\tau}+m\tau}$ and 
$\mathbb{P}(\mathcal{Q}=1)=p=1-\mathbb{P}(\mathcal{Q}=0)$.

When $p=0$, we have $N=\infty$, and $X_N \to X_{\infty}$. 

Let $\alpha$ be the unique positive solution of the equation
\begin{equation}
\varphi(\alpha):=
\mathbb{E}\left[\left(\mathcal{A}(1-\mathcal{Q})\right)^{\alpha}\right]
=\mathbb{E}[\mathcal{A}^{\alpha}]\mathbb{E}[(1-\mathcal{Q})^{\alpha}]=
e^{\kappa(\alpha)\tau+m\alpha\tau}(1-p)=1.
\end{equation}
Then, as $x\rightarrow+\infty$,
\begin{equation}
\mathbb{P}(X_{N}>x)\sim c_{+}x^{-\alpha},
\end{equation}
for some constant $c_{+}>0$. 

Let $Y_{\tau}:=\tau X_{N}$ and $p=\lambda\tau$. 
By a modification of the proof of Theorem \ref{ThmExponentialT}, 
we expect that the sum $Y_{\tau}$ converges to the integral under
some mild conditions. Specifically
\begin{equation}
Y_{\tau}\rightarrow\int_{0}^{T_{\lambda}}e^{Z_{t}+mt}dt, 
\end{equation}
in distribution as $\tau\rightarrow 0$, where $T_{\lambda}$ is exponentially 
distributed with parameter $\lambda>0$, and independent of the L\'{e}vy 
process $Z_{t}$.

The time integral of exponential L\'{e}vy processes has been studied
in the literature, see e.g. Bertoin and Yor \cite{BY} and the references therein.
The discrete approximation can give an alternative derivation of the properties
of the continuous time integral of exponential L\'{e}vy processes.
Note that since $Y_{\tau}=\tau X_{N}$, we have 
\begin{equation}
Y_{\tau} = \mathcal{A} (\tau+(1-\mathcal{Q})Y_{\tau})
\end{equation}
in distribution, where $Y_{\tau}$ is independent of $\mathcal{A,Q}$.
The m.g.f. of $\log Y_{\tau}$ is given by (for $0<\theta<1$)
\begin{equation}\label{LQsub0}
\mathbb{E}[Y_{\tau}^{\theta}]=
\mathbb{E}[\mathcal{A}^{\theta}]\mathbb{E}[(\tau+(1-\mathcal{Q})Y_{\tau})^{\theta}],
\qquad 0<\theta<1,
\end{equation}
where we used the independence of $Y_{\tau}$ and $\mathcal{A}$. 
Note that since $\mathcal{A} = e^{Z_{\tau}+m\tau}$, we have
\begin{equation}\label{LQsub1}
\mathbb{E}[\mathcal{A}^{\theta}]=e^{\theta(\tau)\tau+\theta m\tau}
=1+\kappa(\theta)\tau+\theta m\tau+O(\tau^{2}),
\end{equation}
Following the same arguments as in the geometric Brownian motion case, we get
\begin{equation}
\mathbb{E}[Y_{\tau}^{\theta}]
=\left(1+\kappa(\theta)\tau+\theta m\tau+O(\tau^{2})\right)
\left(\mathbb{E}[Y_{\tau}^{\theta}]+(\theta\mathbb{E}[Y_{\tau}^{\theta-1}]-\lambda\mathbb{E}[Y_{\tau}^{\theta}])\tau
+O(\tau^{1+\theta})\right).
\end{equation}
Let $\tau\rightarrow 0$, the coefficient of $O(\tau)$ term must vanish and therefore
the limit $Y_{0}$ must satisfy the identity
\begin{equation}
\left(\kappa(\theta)+\theta m-\lambda\right)\mathbb{E}[Y_{0}^{\theta}]
+\theta\mathbb{E}[Y_{0}^{\theta-1}]=0,
\end{equation}
which recovers (2.4) in Donati-Martin et al. \cite{DMGY}.

\section{Positive Moments}
\label{Sec:PositiveMoments}

In the general setting of the exponential L\'{e}vy model with geometric 
mortality, the average and higher moments for $X_{N}$ may not exist, for 
an arbitrary drift $m$. For the finitely many positive moments
of $X_{N}$ that do exist, there exists a simple recursion relation to 
compute these positive moments. It is also worth noting that the negative 
moments of $X_{N}$ always exist but do not seem to yield closed-form 
expressions.

Recall that $\mathcal{A}=e^{Z_{\tau}+m\tau}$, where $Z_{t}$ is a L\'{e}vy 
process. 
Thus we have $\mathbb{E}[\mathcal{A}^{k}]=e^{\kappa(k)\tau+mk\tau}<1$
if and only if $\kappa(k)+mk<0$.
In the special case of geometric Brownian motion, 
$\mathbb{E}[\mathcal{A}^{k}]=e^{\frac{1}{2}\sigma^{2}\tau(k^{2}-k)+mk\tau}<1$
if and only if $k<\frac{-2m}{\sigma^{2}}+1$. 

Recall that
\begin{equation}
X_{N} = \mathcal{AQ} + \mathcal{A}(1-\mathcal{Q})(1+X_{N})
\end{equation}
in distribution and $X_{N},\mathcal{A,Q}$ are independent.
Therefore, for any $k\in\mathbb{N}$ such that $\kappa(k)+mk<0$, 
\begin{equation}
\mathbb{E}[X_{N}^{k}]=(1-p)\mathbb{E}[\mathcal{A}^{k}]\mathbb{E}[(1+X_{N})^{k}]
+p\mathbb{E}[\mathcal{A}^{k}],
\end{equation}
which yields the recurrence relation:
\begin{equation}\label{OriginalRecurrence}
\mathbb{E}[X_{N}^{k}]
=\frac{\mathbb{E}[\mathcal{A}^{k}]}{1-(1-p)\mathbb{E}[\mathcal{A}^{k}]}
\left[(1-p)\sum_{j=0}^{k-1}\binom{k}{j}\mathbb{E}[X_{N}^{j}]+p\right],
\qquad
k\in\mathbb{N}, \kappa(k)+mk<0.
\end{equation}

As a first step, let us consider the special case $p=0$, then $N=\infty$ a.s.
and the recurrence relation reduces to
\begin{equation}
\mathbb{E}[X_{\infty}^{k}]
=\frac{\mathbb{E}[\mathcal{A}^{k}]}{1-\mathbb{E}[\mathcal{A}^{k}]}
\sum_{j=0}^{k-1}\binom{k}{j}\mathbb{E}[X_{\infty}^{j}],
\qquad
k\in\mathbb{N}, \kappa(k)+mk<0.
\end{equation}
>From this recurrence relation, we can compute that
\begin{align*}
&\mathbb{E}[X_{\infty}]=\frac{\mathbb{E}[\mathcal{A}]}{1-\mathbb{E}[\mathcal{A}]}\binom{1}{0},
\\
&\mathbb{E}[X_{\infty}^{2}]=
  \frac{\mathbb{E}[\mathcal{A}^{2}]}{1-\mathbb{E}[\mathcal{A}^{2}]}\binom{2}{0}
 +\frac{\mathbb{E}[\mathcal{A}^{2}]}{1
 -\mathbb{E}[\mathcal{A}^{2}]}\binom{2}{1}\frac{\mathbb{E}[\mathcal{A}]}
      {1-\mathbb{E}[\mathcal{A}]}\binom{1}{0},
\\
&\mathbb{E}[X_{\infty}^{3}]=
  \frac{\mathbb{E}[\mathcal{A}^{3}]}{1-\mathbb{E}[\mathcal{A}^{3}]}\binom{3}{0}
+\frac{\mathbb{E}[\mathcal{A}^{3}]}{1-\mathbb{E}[\mathcal{A}^{3}]}
\binom{3}{1}\frac{\mathbb{E}[\mathcal{A}]}{1-\mathbb{E}[\mathcal{A}]}\binom{1}{0}
\\
&\qquad
+\frac{\mathbb{E}[\mathcal{A}^{3}]}{1-\mathbb{E}[\mathcal{A}^{3}]}
   \binom{3}{2}\frac{\mathbb{E}[\mathcal{A}^{2}]}
   {1-\mathbb{E}[\mathcal{A}^{2}]}\binom{2}{0}
+\frac{\mathbb{E}[\mathcal{A}^{3}]}{1
-\mathbb{E}[\mathcal{A}^{3}]}\binom{3}{2}
\frac{\mathbb{E}[\mathcal{A}^{2}]}{1-\mathbb{E}[\mathcal{A}^{2}]}\binom{2}{1}
\frac{\mathbb{E}[\mathcal{A}]}{1-\mathbb{E}[\mathcal{A}]}\binom{1}{0},
\end{align*}
and more generally,
\begin{equation}\label{MidRecurrence}
\mathbb{E}[X_{\infty}^{k}]
=\sum_{k=i_{m}>i_{m-1}>\cdots>i_{1}>i_{0}=0, 1\leq m\leq k}
\prod_{j=1}^{m}\binom{i_{j}}{i_{j-1}}
\frac{\mathbb{E}[\mathcal{A}^{i_{j}}]}{1-\mathbb{E}[\mathcal{A}^{i_{j}}]}.
\end{equation}

Now, let us go back to the original recurrence relation \eqref{OriginalRecurrence}.
Note that we can rewrite \eqref{OriginalRecurrence} as
\begin{equation}
\mathbb{E}[X_{N}^{k}]
=\frac{(1-p)\mathbb{E}[\mathcal{A}^{k}]}{1-(1-p)\mathbb{E}[\mathcal{A}^{k}]}
\left[\sum_{j=0}^{k-1}\binom{k}{j}\mathbb{E}[X_{N}^{j}]+\frac{p}{1-p}\right],
\qquad
k\in\mathbb{N}, \kappa(k)+mk<0.
\end{equation}
>From \eqref{MidRecurrence}, it is not difficult to see that
\begin{align}
\mathbb{E}[X_{N}^{k}]
&=\sum_{k=i_{m}>i_{m-1}>\cdots>i_{1}>i_{0}=0, 1\leq m\leq k}
\prod_{j=2}^{m}\binom{i_{j}}{i_{j-1}}\frac{(1-p)
\mathbb{E}[\mathcal{A}^{i_{j}}]}{1-(1-p)\mathbb{E}[\mathcal{A}^{i_{j}}]}
\\
&\qquad\qquad\qquad\cdot
\frac{(1-p)\mathbb{E}[\mathcal{A}^{i_{1}}]}{1-(1-p)\mathbb{E}[\mathcal{A}^{i_{1}}]}
\left[\binom{i_{1}}{i_{0}}+\frac{p}{1-p}\right]
\nonumber
\\
&=\frac{1}{1-p}\sum_{k=i_{m}>i_{m-1}>\cdots>i_{1}>i_{0}=0, 1\leq m\leq k}
\prod_{j=1}^{m}
\frac{(1-p)\mathbb{E}[\mathcal{A}^{i_{j}}]}{1-(1-p)
\mathbb{E}[\mathcal{A}^{i_{j}}]}\prod_{j=2}^{m}\binom{i_{j}}{i_{j-1}}.
\nonumber
\end{align}

\section{Applications to Annuities and Asian Options}\label{ApplicationSection}

In this section, we consider the applications of our results to annuities
and Asian options. As an illustration, we discuss only the case of the sum of 
geometric Brownian motions, so that the model for the Asian options is the 
standard Black-Scholes model. 
It is worth noting that all the discussions in Section \ref{ApplicationSection}
are valid for the sum of exponential L\'{e}vy processes as well.

\subsection{Annuities with Finite Mortality}

We have already analyzed the annuities with geometric distributed mortality.
Now, let us turn to the annuities with finite mortality $n$, and we are 
interested to compute the cumulative distribution function of 
$X_{n}=\sum_{i=1}^{n}e^{\sigma W_{t_{i}}+(m-\frac{1}{2}\sigma^{2})t_{i}}$, 
that is, for any $x>0$: the value of $\mathbb{P}(X_{n}\leq x)$.

For any $0<z<1$, we can compute that
\begin{equation}
G(z)=\sum_{n=1}^{\infty}\mathbb{P}(X_{n}\leq x)z^{n}
=\frac{z}{1-z}\mathbb{P}(X_{N_{p}}\leq x),
\nonumber
\end{equation}
where $N_{p}$ has a geometric distribution with $p=1-z$. Then, 
\begin{align}\label{XNcdf}
\mathbb{P}(X_{n}\leq x)
&=\frac{1}{n!}\frac{d^{n}}{dz^{n}}G(z)\bigg|_{z=0}
\\
&=\frac{1}{n!}\sum_{k=0}^{n}\binom{n}{k}\left(\frac{z}{1-z}\right)^{(n-k)}\bigg|_{z=0}
\frac{d^{k}}{dz^{k}}\mathbb{P}(X_{N_{p}}\leq x)\bigg|_{z=0}
\nonumber
\\
&=\sum_{k=0}^{n-1}\frac{1}{k!}
\frac{d^{k}}{dz^{k}}\mathbb{P}(X_{N_{p}}\leq x)\bigg|_{z=0}
\nonumber
\\
&=\sum_{k=0}^{n-1}\frac{1}{k!}(-1)^{k}
\int_{0}^{x}\frac{\partial^{k}}{\partial p^{k}}f(y;\beta,\rho,p)dy\bigg|_{p=1},
\nonumber
\end{align}
where we recall that $f(x;\beta,\rho,p)$ is the probability density function of $X_{N}$
with $\beta=\sigma^{2}\tau$ and $\rho=m\tau$.

We give next a recursive representation for the coefficients in this expansion
expressed in terms of the distribution function $f(x;\beta,\rho,p)$
and its derivatives with respect to $p$ at $p=1$.
\begin{theorem}\label{Thm22}
The finite sum $X_{n}=\sum_{i=1}^{n}e^{\sigma W_{t_{i}}+(m-\frac{1}{2}\sigma^{2})t_{i}}$ 
has the probability density function
\begin{equation}\label{Xnsum}
f_n(x;\beta,\rho) = \sum_{k=0}^{n-1}\frac{1}{k!}(-1)^{k}
\frac{\partial^{k}}{\partial p^{k}}f(x;\beta,\rho,1),
\end{equation}
where 
\begin{align}\label{deriv0}
&f(x;\beta,\rho,1)=\frac{1}{\sqrt{2\pi\beta}x}
e^{-\frac{1}{2\beta}(\log x-\rho+\frac{1}{2}\beta)^{2}},
\\
\label{deriv1}
&\frac{\partial}{\partial p}f(x;\beta,\rho,1)=\frac{1}{\sqrt{2\pi\beta}x}
e^{-\frac{1}{2\beta}(\log x-\rho+\frac{1}{2}\beta)^{2}}
-\frac{1}{x} \int_0^\infty
e^{-\frac{1}{2\beta} 
\left[ \log \frac{x}{1+y} + \frac12 \beta - \rho \right]^2}
\frac{e^{-\frac{1}{2\beta}(\log y-\rho+\frac{1}{2}\beta)^{2}}}{2\pi\beta y}dy,
\end{align}
and for any $k\geq 2$,
\begin{eqnarray}
\frac{\partial^{k}}{\partial p^{k}}f(x;\beta,\rho,1)= 
-k\frac{1}{x} \int_0^\infty \frac{dy}{\sqrt{2\pi\beta}} 
\exp\left( -\frac{1}{2\beta} 
\left[ \log \frac{x}{1+y} + \frac12 \beta - \rho \right]^2
\right)
\frac{\partial^{k-1}}{\partial p^{k-1}}f(y;\beta,\rho,1)\,.\nonumber \\
\label{fderivk}
\end{eqnarray}
\end{theorem}
\begin{proof}

Recall from Proposition \ref{prop:14} that the density function of $X_N$ 
satisfies the integral equation (\ref{IntegralEqXN}).
By letting $p=1$ in this equation, we get equation \eqref{deriv0}.
Differentiating equation \eqref{IntegralEqXN} with respect to $p$
and setting $p=1$, we get
\begin{align}
&\frac{\partial}{\partial p}f(x;\beta,\rho,p)\bigg|_{p=1}
\\
&= 
\frac{1}{\sqrt{2\pi\beta}x}
e^{-\frac{1}{2\beta}(\log x-\rho+\frac{1}{2}\beta)^{2}}
\nonumber
\\
&\qquad
-\frac{1}{x} \int_0^\infty \frac{dy}{\sqrt{2\pi\beta}} 
\exp\left( -\frac{1}{2\beta} 
\left[ \log \frac{x}{1+y} + \frac12 \beta - \rho \right]^2
\right)
f(y;\beta,\rho,1).
\nonumber
\end{align}
which reproduces \eqref{deriv1}.
Moreover, for any $k\in\mathbb{N}$
and $k\geq 2$, differentiating equation \eqref{IntegralEqXN}
$k$ times with respect to $p$ and setting $p=1$, we get the
equation \eqref{fderivk}.
This gives $\frac{\partial^{k}}{\partial p^{k}}f(x;\beta,\rho,1)$ for every 
$k=0,1,2,\ldots$. Substitution into (\ref{XNcdf}) and taking one derivative
with respect to $x$ gives the representation
(\ref{Xnsum}) for the density of $X_n$. 
This concludes the proof of this relation.

\end{proof}

The finite sum $X_n$ satisfies the recursion
\begin{equation}
X_n = \mathcal{A} (1 + X_{n-1})\,,
\end{equation}
where $\mathcal{A}$ is defined in (\ref{Adef}). This gives a recursive relation
for the density of $X_n$ which can be written in symbolic form as
\begin{equation}\label{rec}
f_n(x;\beta,\rho) = \hat T_{\beta,\rho} f_{n-1}(x;\beta,\rho) \,,
\end{equation}
where $\hat T_{\beta,\rho}$ denotes the integral transform in (\ref{integraleq}),
with initial condition $f_1(x;\beta,\rho) = f(x;\beta,\rho,1)$.
This is solved formally as 
\begin{equation}\label{Xnprod}
f_n(x;\beta,\rho) = \hat T_{\beta,\rho}^{n-1} f_{1}(x;\beta,\rho) \,.
\end{equation}

\begin{remark}
Theorem~\ref{Thm22} gives an explicit additive solution for the recursion
(\ref{rec}). In order to see this we note that the terms appearing in 
(\ref{Xnsum}) can be written alternatively as 
\begin{equation}
\frac{\partial^k}{\partial p^k} f(x;\beta,\rho,1) = (-1)^{k-1} k! \hat T_{\beta,\rho}^{k-1}
(1 - \hat T_{\beta,\rho}) f_1(x;\beta,\rho) \,.
\end{equation}
It is easy to see by substitution into (\ref{Xnsum}) that the total
result agrees with (\ref{Xnprod}).
\end{remark}

We can also study the left tails and right tails
of the finite sum of geometric Brownian motions:
\begin{proposition}\label{Prop:XnLT}
For any $n\in\mathbb{N}$, we have
\begin{equation}
\lim_{\epsilon\rightarrow 0}\frac{\log\mathbb{P}(X_{n}\leq\epsilon)}{(\log\epsilon)^{2}}=-\frac{1}{2\sigma^{2}\tau}.
\end{equation}
\end{proposition}

\begin{proof}
Note that for any $n\in\mathbb{N}$, $X_{1}\leq X_{n}\leq X_{\infty}$. 
Since $X_{1}$ is log-normally distributed, it is clear that
$\lim_{\epsilon\rightarrow 0}\frac{\log\mathbb{P}(X_{1}\leq\epsilon)}{(\log\epsilon)^{2}}=-\frac{1}{2\sigma^{2}\tau}$.
Then, the result follows from Proposition \ref{LeftTailProp}.
\end{proof}

We have the following estimate for the right tail asymptotics:
\begin{proposition}
For any $n \in \mathbb{N}$ we have
\begin{equation}\label{XnRT}
\lim_{x\rightarrow\infty}\frac{\log\mathbb{P}(X_{n}\geq x)}{(\log x)^{2}}
= - \frac{1}{2\sigma^2\tau n}\,.
\end{equation}
\end{proposition}

\begin{proof}
We prove matching upper and lower bounds. We start by deriving an upper 
bound, which follows by writing
\begin{eqnarray}
 X_n \leq \sum_{k=1}^n 
   e^{\sigma \max_{1 \leq i \leq n} W_{t_i} + |m - \frac12\sigma^2|t_n}  
= n e^{\sigma \max_{1 \leq i \leq n} W_{t_i} + |m - \frac12\sigma^2|t_n}.
\end{eqnarray}
By the reflection principle, $\max_{1 \leq k \leq n} W_{t_k} = |W_{t_n}|$ in 
distribution. This gives
\begin{align}
\mathbb{P}(X_n > x) 
&\leq\mathbb{P}\left(e^{\sigma |W_{t_n}|} > 
\frac{x}{n} e^{-|m-\frac12\sigma^2|t_n} \right)
\\
& = 2 \mathbb{P}\left(
Z > \frac{1}{\sigma\sqrt{t_n}} \left(  \log(x/n) - |m-\frac12\sigma^2| t_n
\right) \right) 
\nonumber 
\\
& = 2\Phi \left( - \frac{1}{\sigma\sqrt{t_n}} 
   \left( \log(x/n) - |m-\frac12\sigma^2| t_n \right) \right) 
\nonumber 
\\
& \leq 2 \frac{\sigma\sqrt{t_n}}{\sqrt{2\pi} L} 
   e^{-\frac{1}{2\sigma^2 t_n}L^2},
\nonumber
\end{align}
where we denoted $L=\log(x/n) - |m-\frac12\sigma^2|t_n$ 
and $\Phi(x):=\frac{1}{\sqrt{2\pi}}\int_{-\infty}^{x}e^{-\frac{y^{2}}{2}}dy$ is 
the cumulative distribution function of $N(0,1)$.
Here $Z=N(0,1)$ and we used in the last line the inequality 
\begin{equation}\label{Nbounds}
\frac{1}{\sqrt{2\pi} x} e^{-\frac12 x^2} \left(1 - \frac{1}{x^2}\right) \leq
\Phi(-x) \leq \frac{1}{\sqrt{2\pi} x} e^{-\frac12 x^2}\,, \quad x>0\,.
\end{equation}
Taking the logs of both sides, dividing by $\log^2 x$ and taking the
$x\to \infty$ limit gives
\begin{equation}
\limsup_{x\to \infty} \frac{\log\mathbb{P}(X_n>x)}{\log^2 x} \leq 
   - \frac{1}{2\sigma^2\tau n}.
\end{equation}
This proves the upper bound for (\ref{XnRT}). 

Next we prove a matching lower bound. This is obtained from the inequality 
\begin{equation}
X_n > e^{\sigma W_{t_n} + (m-\frac12\sigma^2)t_n},
\end{equation}
which implies
\begin{align}
& \mathbb{P}(X_n > x) > \mathbb{P}\left(e^{\sigma W_{t_n}} > 
  x e^{-(m-\frac12\sigma^2)t_n}\right) 
\\
&= \mathbb{P}\left( Z > \frac{1}{\sigma \sqrt{t_n}} 
                 \log\left(x e^{-(m-\frac12\sigma^2)t_n}\right)\right) \nonumber \\
&\geq \frac{\sigma\sqrt{t_n}}{\sqrt{2\pi} L'} 
e^{-\frac{1}{2\sigma^2 t_n}L'^2} \left(1 - \frac{\sigma^2 t_n}{L'^2}\right),\nonumber
\end{align}
where we denoted $L'=\log x - (m-\frac12\sigma^2) t_n$ and used again the 
inequality (\ref{Nbounds}). Taking the logs of both sides, dividing by
$\log^2 x$ and taking the $x\to \infty$ limit gives
\begin{equation}
\liminf_{x\to \infty} \frac{\log\mathbb{P}(X_n>x)}{\log^2 x} \geq
    - \frac{1}{2\sigma^2 \tau n} \,.
\end{equation}
This proves the lower bound.
This completes the proof of (\ref{XnRT}).  
\end{proof}

\begin{remark}
The right tail asymptotics of the discrete sum of GBM (\ref{XnRT}) is similar 
to the right tail asymptotics of the time integral of the GBM which was studied 
in \cite{LingExtremeStrikes} in relation to the large strike asymptotics of the
out of money Asian call options in the Black-Scholes model. From Proposition 1(i) 
in \cite{LingExtremeStrikes}
one finds
\begin{equation}
\lim_{x\to \infty}\frac{\mathbb{P}(\int_0^T dt e^{\sigma W_t + (r-\frac12\sigma^2)t} > x)}{\log^2 x}
= - \frac{1}{2\sigma^2 T}\,.
\end{equation}
A similar result is obtained for the left tail asymptotics of the time integral
of the GBM (Proposition 1(ii) in \cite{LingExtremeStrikes}). 
The corresponding asymptotics for the left tail of the sum of GBM is however
different, as seen from Proposition~\ref{Prop:XnLT}.
\end{remark}

The tail asymptotics of the sum of correlated log-normal random variables 
has been widely studied in the literature \cite{ARN,GXY,GT}, see
\cite{reviewLN} for a review of the literature and applications.
The right tail asymptotics of the sum of correlated log-normal random variables
has been completely characterized in \cite{ARN}. Our result 
(\ref{XnRT}) agrees with the results of \cite{ARN}, specialized to the sum of 
GBM. The asymptotics of the left tail has been recently also studied in
\cite{GT} for an arbitrary number of log-normal variables, and for $n=2$ in
\cite{GXY}. However, the results of \cite{GT} are obtained under a certain 
assumption (denoted Assumption $\mathcal{A}$ in \cite{GT}) which does not hold 
for the sum of GBM, such  that their results cannot be applied to our problem.

\subsection{Annuities with Stochastic Mortality}

Now assume that the mortality time $N$ has a general distribution:
\begin{equation}
\mathbb{P}(N=n)=p_{n},\qquad n=1,2,3,\ldots,
\end{equation}
and $N$ is independent of the geometric Brownian motion. 

We discussed the distribution of $X_{N}$, with $N$ following a geometric 
distribution. Using this result, we also derived the distribution of $X_{n}$ 
for a finite given $n$. 
When $N$ follows a general distribution, we denote the corresponding sum of
GBM as $X_R$. The cumulative distribution function 
of $X_{R}$ is given, for any $x>0$, by
\begin{equation}
\mathbb{P}(X_{R}\leq x)=\sum_{n=1}^{\infty}p_{n}\mathbb{P}(X_{n}\leq x),
\end{equation}
and $f_{R}(x)$, the probability density function of $X_{R}$, is thus given by
\begin{equation}
f_{R}(x)=\sum_{n=1}^{\infty}p_{n}f_{n}(x),
\end{equation}
where $f_{n}(x)$ is the probability density function of $X_{n}$.

To summarize, for the general stochastic mortality annuities, we can use
geometric mortality to derive the distribution for the finite mortality
and then use this to finally obtain the distribution for the general stochastic 
mortality. 

An alternative method was proposed in \cite{GSY} where it was showed that 
any positive-definite discrete distribution can be matched arbitrarily close 
by an appropriate linear combination of geometric distributions. This is the
discrete time counterpart of a continuous-time result \cite{Dufresne2007}, 
that states that any positive definite continuous distribution can be 
approximated arbitrarily close by an appropriate linear combination of 
exponential distributions. 

\subsection{Risk Measures of Annuities}

In practical applications one is interested in the 
probability that the annuity $S_n$ exceeds a certain value $K$,
giving the available amount from which the cash flows are paid.
This defines the shortfall probability $\mathbb{P}(S_n > K)$.

We will compute in this section the shortfall probability of the
sum of geometric Brownian motion with a geometrically distributed
stopping time $X_N$. 
Using the right-tail asymptotics derived in equation (\ref{XNrighttail}), 
this is given
for $K \to \infty$ by the complementary cumulative distribution function
\begin{equation}
\mathbb{P}(X_N > K) = \int_K^\infty dx f(x;\beta,\rho,p) \sim
c_{+}K^{-\mu},
\end{equation}
with $\mu$ given by (\ref{mugeom}), and $c_+>0$ is a positive constant
determined by (\ref{cplus}).

An alternative risk metric is the Value-at-Risk which is defined as
that  amount $K$
for which  the probability of $X_N$ exceeding $K$ takes a known value,
e.g. 5\% or 1\%. We define thus
\begin{equation}
\mbox{p-VaR} =\inf\left\{K\geq 0:\mathbb{P}(X_{N}>K)\geq p\right\}.
\end{equation}
Using again the right tail asymptotics for $X_N$ we have
\begin{equation}
\mbox{p-VaR} = - \frac{1}{\mu} \log(p/c_+),
\end{equation}
for sufficiently small $p$, 
with $\mu$ given by (\ref{mugeom}), and $c_+>0$ given in (\ref{cplus}).

In practical applications the distributional properties of the discrete time 
annuities are likely to be studied using numerical methods, such as Monte Carlo
simulations. Such methods are known to be unreliable for sampling the tail
probabilities as they require very long simulation times \cite{reviewLN}.
Using the exact tail behavior obtained in this paper it is possible to
obtain reliable risk metrics for the shortfall probabilities of discrete time
annuities, and to construct efficient simulation methods. Another possible
approach is to use continuous time approximations to study the distribution
of the annuities with exponential mortality, for which detailed theoretical
results are available, see \cite{DufresneAnnuity}.
In Section~\ref{NumericalSection} we will study the impact of the continuous 
time approximation for the distributional properties of discrete time annuities.


\subsection{Applications to Asian Options}

We have studied the distribution of
\begin{equation}
X_{N}=\sum_{i=1}^{N}e^{\sigma W_{t_{i}}+(m-\frac{1}{2}\sigma^{2})t_{i}},
\end{equation}
where $N$ follows a geometric distribution, independent of the Brownian motion $W_{t}$, that is,
\begin{equation}
\mathbb{P}(N=k)=(1-p)^{k-1}p,
\qquad
k=1,2,3,\ldots.
\end{equation}

Let $m=r-q$, where $r$ is the risk-free rate and $q$ is the dividend yield. 
Then the Asian call option price with strike price $K>0$
and initial stock price $S_{0}>0$ for the Black-Scholes model
is given by
\begin{equation}
C=e^{-r\tau n}\mathbb{E}\left[\left(\frac{1}{n}\sum_{i=1}^{n}S_{0}e^{\sigma W_{t_{i}}+(m-\frac{1}{2}\sigma^{2})t_{i}}-K\right)^{+}\right].
\end{equation}
Therefore, to compute the call option price, it suffices to compute:
\begin{equation}
P_{n}:=\mathbb{E}\left[\left(\sum_{i=1}^{n}e^{\sigma W_{t_{i}}+(m-\frac{1}{2}\sigma^{2})t_{i}}-\kappa\right)^{+}\right],
\end{equation}
for any positive number $\kappa>0$.

Assume for any $0<z<1$, we can compute the generating function of $P_{n}$, that is,
\begin{equation}
F(z):=\sum_{n=1}P_{n}z^{n}.
\end{equation}
Then, it is clear that $P_{n}$ can be computed as the $n$-th derivative of $F(z)$ w.r.t. $z$ at $z=0$, that is,
\begin{equation}
P_{n}=\frac{1}{n!}\frac{d^{n}}{dz^{n}}F(z)\bigg|_{z=0}.
\end{equation}

On the other hand, it is easy to see that
\begin{align}
F(z)&=\sum_{n=1}^{\infty}\mathbb{E}\left[\left(\sum_{i=1}^{n}e^{\sigma W_{t_{i}}+(m-\frac{1}{2}\sigma^{2})t_{i}}-\kappa\right)^{+}\right]z^{n}
\\
&=\frac{z}{1-z}\sum_{n=1}^{\infty}\mathbb{E}\left[\left(\sum_{i=1}^{n}e^{\sigma W_{t_{i}}+(m-\frac{1}{2}\sigma^{2})t_{i}}-\kappa\right)^{+}\right](1-(1-z))^{n-1}(1-z)
\nonumber
\\
&=\frac{z}{1-z}\mathbb{E}\left[\left(X_{N}-\kappa\right)^{+}\right],
\nonumber
\end{align}
with $p=1-z$ in the definition of $X_{N}$, where $N$ is geometrically distributed with parameter $p$
and we have already discussed the properties of the distribution of $X_{N}$ in the previous sections.
Therefore, 
\begin{align}
P_{n}&=\frac{1}{n!}\sum_{k=0}^{n}\binom{n}{k}\left(\frac{z}{1-z}\right)^{(n-k)}\bigg|_{z=0}
\frac{d^{k}}{dz^{k}}\mathbb{E}\left[\left(X_{N}-\kappa\right)^{+}\right]\bigg|_{z=0}
\\
&=\frac{1}{n!}\sum_{k=0}^{n}\binom{n}{k}(n-k)!
\frac{d^{k}}{dz^{k}}\mathbb{E}\left[\left(X_{N}-\kappa\right)^{+}\right]\bigg|_{z=0}
\nonumber
\\
&=\sum_{k=0}^{n-1}\frac{1}{k!}
\frac{d^{k}}{dz^{k}}\mathbb{E}\left[\left(X_{N}-\kappa\right)^{+}\right]\bigg|_{z=0}
\nonumber
\\
&=\sum_{k=0}^{n-1}\frac{1}{k!}(-1)^{k}
\int_{\kappa}^{\infty}\frac{\partial^{k}}{\partial p^{k}}(x-\kappa)f(x;\beta,\rho,p)\bigg|_{p=1}.
\nonumber
\end{align}

Similarly one can compute the price of Asian put options and any Asian 
type options with payoff being a function of $X_{n}$. 

Finally, we remark that in the continuous time setting, one
can use the exponentially distributed maturity Asian options
and then use the inverse Laplace transform to obtain
the Asian option prices with finite maturity, as shown in
Geman, Yor \cite{GemanYor} and Carr, Schr\"oder \cite{CS}. 
See \cite{DufresneReview} for a review. 
Our approach using geometrically distributed maturity
for the discrete time Asian options, is the discrete time analogue
of the continuous time approach familiar from the literature.


\section{Numerical Studies}\label{NumericalSection}

\subsection{Infinite Sum of the GBM}

We will compute the density function $f(x;\beta, \rho)$ of $X_\infty$
by solving the integral equation (\ref{integraleq}). 
The numerical evaluation of the integral is simplified by introducing
the new variable $u = \log(x+1)$ taking values in $u:(0,\infty)$,
and the new unknown function $F(u;\beta,\rho) = f(e^u-1;\beta,\rho)$. 
With this change of variables the equation (\ref{integraleq}) becomes
\begin{equation}\label{integraleq2}
F(u;\beta,\rho) = e^{\beta-\rho}
\int_0^\infty \frac{dw}{\sqrt{2\pi\beta}}
e^{-\frac{1}{2\beta}(w - w_0(u))^2} F(w;\beta,\rho),
\end{equation}
with $w_0(u) = \log(e^u-1) + \frac32 \beta - \rho$. This eliminates the
factor of $1/x$ in (\ref{integraleq}) which could introduce numerical noise 
for small values of $x$. We used trapezoidal quadrature with step
$h = 0.01$. The convergence of the trapezoidal quadrature as
$h \to 0$ is controlled by the following theorem (\cite{DR}, page 208). 
\begin{theorem}\label{Thm:Eh}
Let $a$ and $k$ be fixed, and let $f(x) \in C^{2k+1}[a,b]$ for all $b\geq a$. 
Suppose further that $\int_a^\infty dx f(x)$ exists, that
\begin{equation}
M = \int_a^\infty |f^{(2k+1)}(x)| dx \leq \infty,
\end{equation}
and that 
\begin{eqnarray}
&& f'(a) = f^{(3)}(a) = \cdots = f^{(2k-1)}(a) = 0, \\
&& \lim_{x\to \infty} f'(x) = 
   \lim_{x\to \infty} f^{(3)}(x) = \cdots = 
   \lim_{x\to \infty} f^{(2k-1)}(x) = 0 \,.
\end{eqnarray}
Then the quadrature error for trapezoidal quadrature with step $h>0$ is bounded
from above as
\begin{eqnarray}
&& E_h = \left|\int_a^\infty f(x) dx - 
      h\left[ \frac12 f(a) + f(a+h) + f(a+2h) + \cdots \right]\right| \\
&& \qquad \leq h^{2k+1} \frac{M\zeta(2k+1)}{2^{2k} \pi^{2k+1}}, \nonumber
\end{eqnarray}
where $\zeta(p)=\sum_{j=1}^\infty j^{-p}  $ is the Riemann zeta function.
\end{theorem}
The conditions of this theorem are satisfied by the integrand in (\ref{integraleq2}) 
for any $k \geq 0$. The function $F(u;\beta,\rho)$ and all its derivatives
vanish at $u=0$, as seen from equation \eqref{upperEp}. The right tail asymptotics
proves that all the derivatives vanish at $x\to \infty$ as well. We will use
this theorem to estimate an upper bound on the quadrature error by numerical
evaluation of the constant $M$.

We will solve the equation (\ref{integraleq2}) by 
iteration, starting with an initial function $f_0(x;\beta,\rho)$ on the
right hand side, and using  the result $f_1(x;\beta,\rho)$ as integrand
in the next step. We stop when convergence is reached, to a prescribed
degree of accuracy, as measured by the $L_\infty$ norm of the difference
between successive iterations
\begin{equation}
\Delta_n = \Vert f_n(x) - f_{n-1}(x)\Vert_{L_\infty}.
\end{equation}
As an illustration of the rate of convergence of the iteration we show in 
Figure~\ref{Fig:1} (left plot) plots of $\log_{10}|\Delta_n|$ vs $n$ 
for the iteration of equation (\ref{integraleq2}) with $\rho = -0.1$ and
several values of $\beta =1, 0.5, 0.1, 0.01$, for an initial
condition given by the inverse Gamma distribution with the appropriate $\beta,\rho$
parameters. For these cases the error approaches about $10^{-8}$ after about 
$\sim 100$ iterations.  We checked also that the normalization
of the density function $f(x;\beta,\rho)$ is correctly preserved to 1 
during the iteration. The quadrature error was estimated using 
Theorem~\ref{Thm:Eh} with $k=1$, and numerical evaluation of the constant $M$.
For $\beta=1,\rho=0$ this gives $E_h \leq 0.41 h^3$ and for $\beta=0.1,\rho=-0.1$
we have $E_h \leq 0.058 h^3$. We used $h=0.01$ such that the quadrature error 
is below $10^{-6}$ in all cases considered.

As starting function for the iteration we used two choices:

i) $f_0(x) = \phi_\infty(x;\rho,\beta)$. This is the inverse Gamma 
distribution, giving the distribution of $\tau X_\infty$ in the small $\tau$
limit.

ii) $f_0(x) = \frac{1}{\sqrt{2\pi\beta}x}e^{-\frac{1}{2\beta} 
(\log x + \frac12\beta-\rho)^2}$ the log-normal distribution of the
multiplier $\mathcal{A}$.
We checked that the iteration converges to the same distribution for both 
initial distributions.

In Figure~\ref{Fig:2} we show the density of $X_\infty$ given by 
$F(u;\beta,\rho)$ for $\beta=1$
and $\beta=0.1$, comparing the solution of the integral equation (solid curves)
with the continuous time approximation given by 
Proposition~\ref{prop:smalltau} (dashed curves). As 
expected from (\ref{limf2}), the density function approaches the 
inverse Gamma distribution $F(u;\beta,\rho)\to \phi_\infty(e^u-1;\beta,\rho)$ 
as $\beta,\rho \to 0$. The discrete time distribution $F(u;\beta,\rho)$ is more 
concentrated near the origin, and the right tail is more suppressed than the 
inverse Gamma distribution which is the continuous time limit. 

\begin{figure}[t]
\centering
\includegraphics[width=3in]{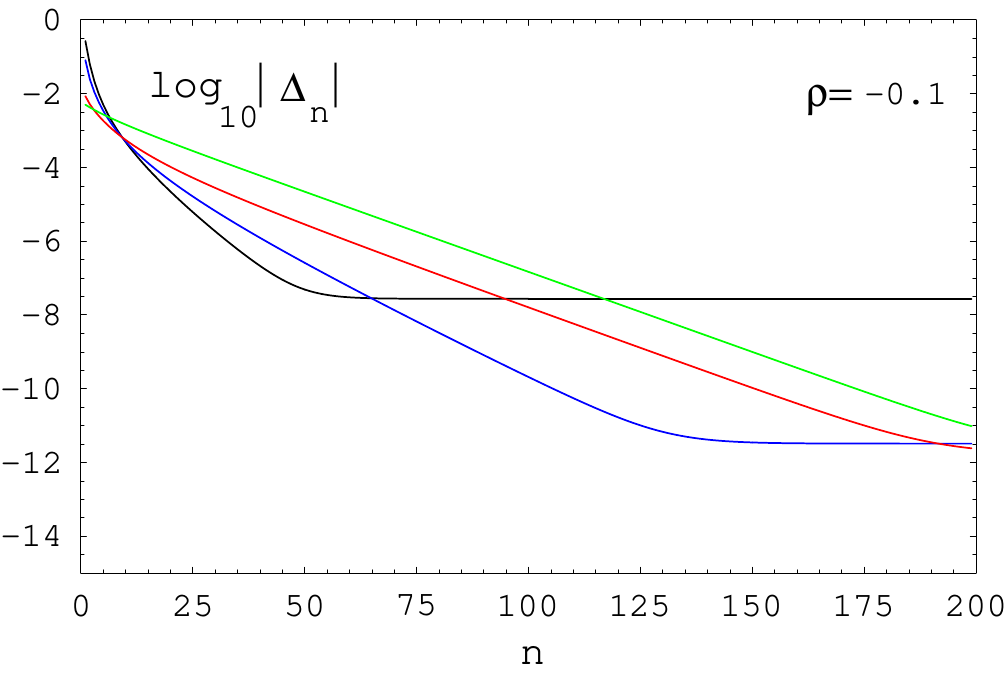}
\includegraphics[width=3in]{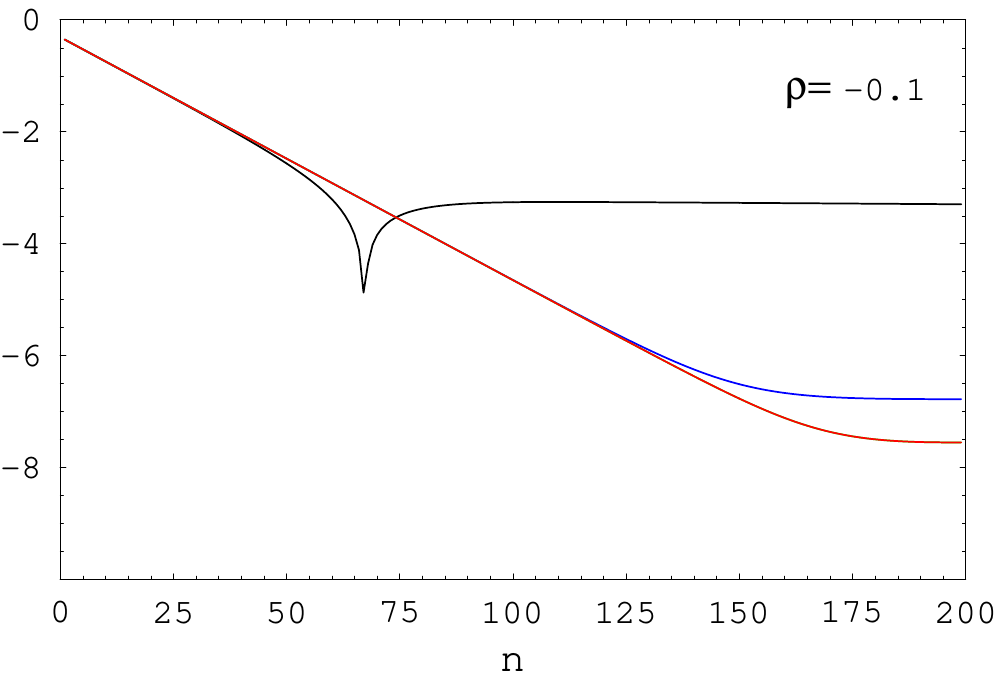}
\caption{
Left: Plot of $\log_{10}|\Delta_n|$ vs $n$ for the
iteration of equation (\ref{integraleq2}) with parameters $\rho = -0.1$
and $\beta=1$ (black), 0.5 (blue), 0.1 (red), 0.01 (green), and initial
condition $F_0$ given by the inverse Gamma distribution.
Right: plot of $\log_{10}|\mu_1(n) - \frac{e^\rho}{1-e^\rho}|$
vs $n$, with $\mu_1(n)$ the estimate of the first moment of $X_\infty$ 
after $n$ iterations. Same parameters as for the left plot.}
\label{Fig:1}
\end{figure}

These plots show also the impact of the $\rho$ parameter on the shape of the 
density function $F(u;\beta,\rho)$.
As expected, negative values of $\rho$ increase the density at small 
values of $X_\infty$, while positive values depress it, and increase the
contribution of the right tail.

As a test for the quality of the numerical solution we computed the
first moment $\mathbb{E}[X_\infty]$ using the density function 
$f(x;\beta,\rho)$. As discussed in Section~\ref{Sec:PositiveMoments}, 
this moment is finite for $\rho < 0$ and is given by 
$\mu_1=\mathbb{E}[X_\infty] = \frac{e^\rho}{1-e^\rho}$.
For $\rho=-0.1$, the convergence of the first moment to the theoretical
value is shown in Figure~\ref{Fig:1} (right) which shows plots of
$\log_{10}|\mu_1(n) - \frac{e^\rho}{1-e^\rho}|$ vs $n$ for several
values of $\beta$, with $\mu_1(n) = \int_0^\infty x f_n(x;\beta,\rho) dx$.
The spike in the $\beta=1$ plot is due to a change of sign of  the difference
$\mu_1(n) - \frac{e^\rho}{1-e^\rho}$.

The numerical values of the $(\beta,\rho)$ parameters used for these 
simulations cover the range of realistic values corresponding to practical
applications.
Typical values used in simulations for the equities model are
$\mu = 0.06, \sigma =0.2$ \cite{DufresneAnnuity}. Assuming
$r=0.01-0.05$ for the deterministic discount rate, gives that 
$m = -\mu - r$ takes values between -0.07 and -0.11. 

Assuming a annuity with monthly payments $\tau = 1/12$ the parameters 
determining the shape of the annuity density distribution are $\beta = 0.0033, 
-\rho = 0.006 - 0.009$. For a yearly annuity $\tau=1$ the parameters are 
$\beta = 0.04, -\rho=0.07-0.11$. As seen from Figure~\ref{Fig:1} the 
convergence properties of the iterative procedure for solving the equation
(\ref{integraleq2}) for these parameter values are good, and the iteration
converges with a moderate number of iterations of the order $n\sim 100$.

\begin{figure}[t]
\centering
\includegraphics[width=4in]{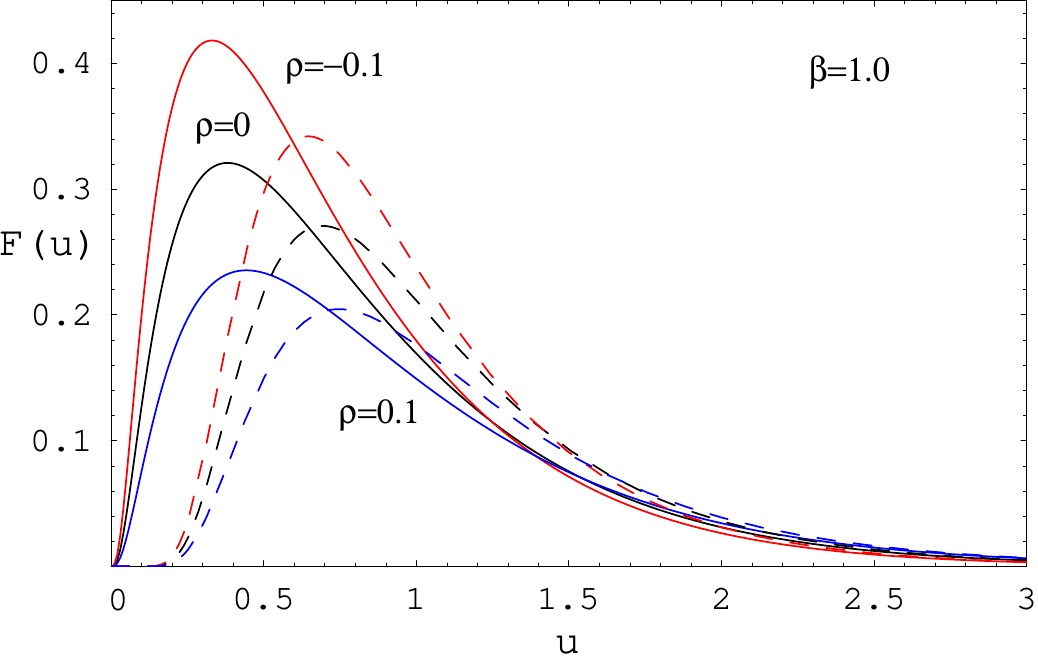}
\includegraphics[width=4in]{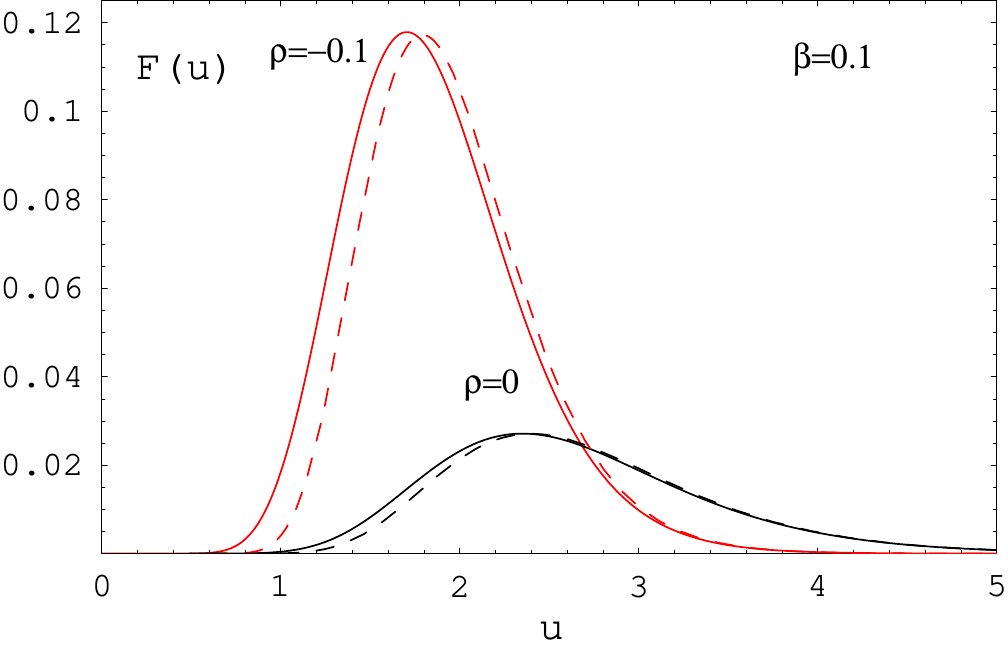}
\caption{Plots of $F(u;\beta,\rho)$ (solid curves), the density function of $X_\infty$,
for several values of $(\beta,\rho)$.
The dashed curves show the continuous time approximation
given by Proposition~\ref{prop:smalltau}.
Above: $\beta=1$, $\rho =-0.1$ (red), $\rho=0$ (black)
and $\rho=0.1$ (blue). 
Below: $\beta=0.1$, $\rho =-0.1$ (red), $\rho=0$ (black).}
\label{Fig:2}
\end{figure}

\begin{figure}[t]
\centering
\includegraphics[width=3in]{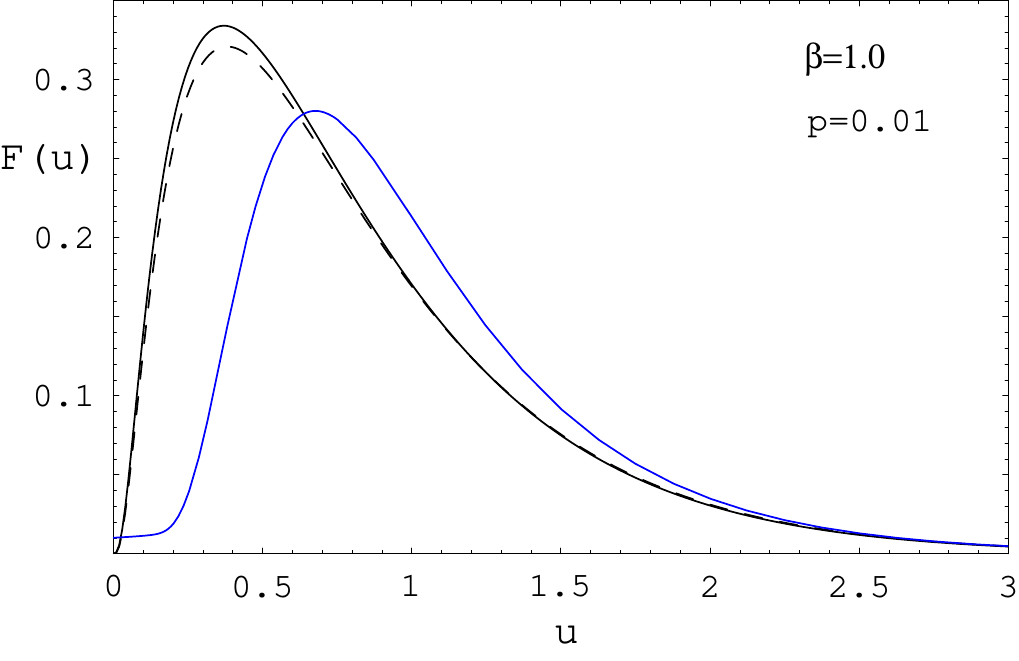}
\includegraphics[width=3in]{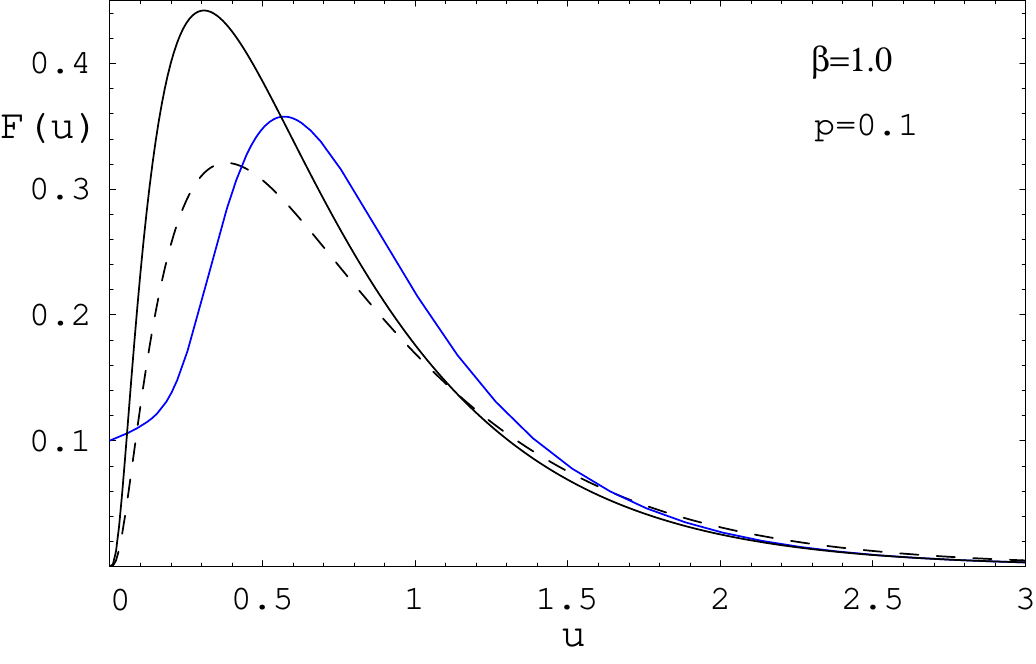}
\includegraphics[width=3in]{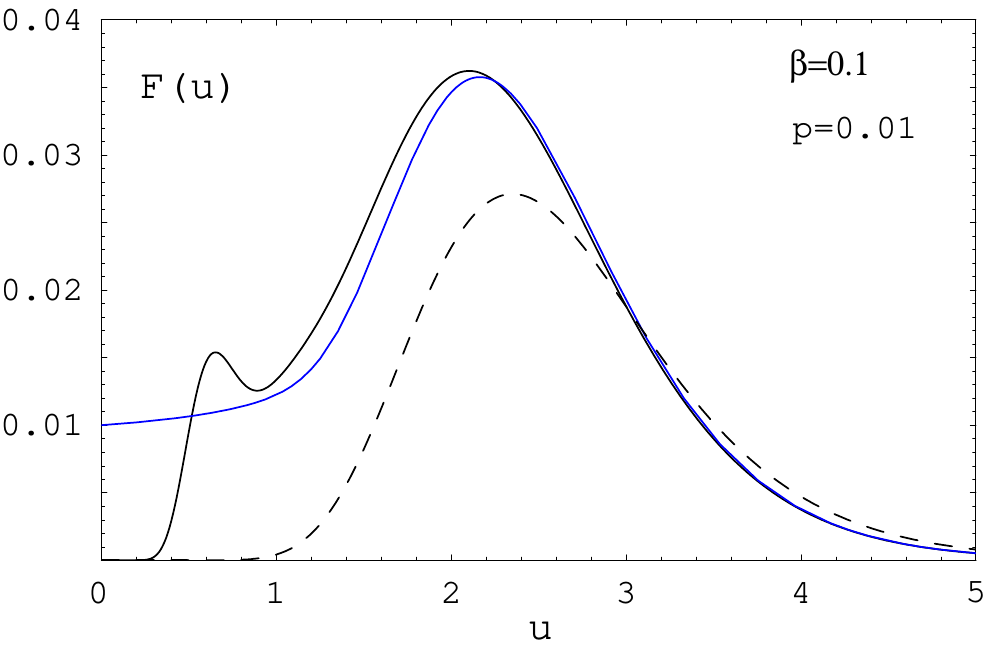}
\includegraphics[width=3in]{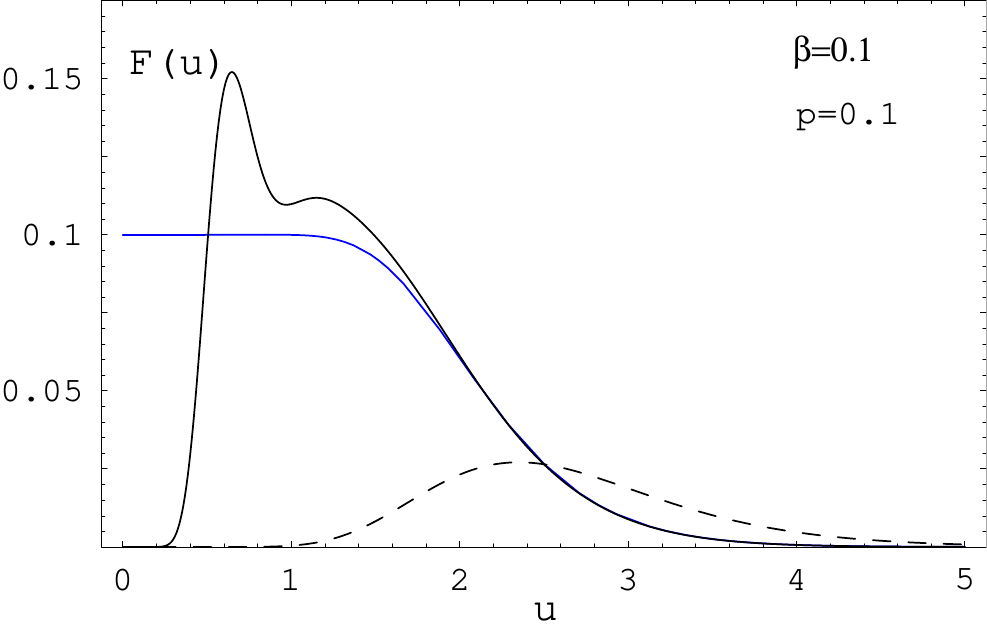}
\caption{
Plots of $F(u;\beta,\rho,p)$ (black curves), the density of $X_N$, for several
choices of the parameters.
Above: $\beta=1,\rho=0$, below: $\beta = 0.1, \rho=0$.
Left: $p= 0.01$, right: $p=0.1$.
The dashed curves correspond to $p=0$, and the blue curves show the 
continuous time approximation given by (\ref{fNsmalltau}).}
\label{Fig:4}
\end{figure}

\subsection{Annuities with Geometric Mortality}

We consider here the distributional properties of the annuity with
geometric mortality. 
Under the geometric mortality model, the density function 
$f(x;\beta,\rho,p)$ of the annuity $X_N$ is found by solving the integral 
equation in Proposition~\ref{prop:14}. 
After the change of variable $u=\log(x+1)$ this gives an integral
equation for the function $F(u;\beta,\rho,p)=f(e^u-1;\beta,\rho,p)$.
\begin{align}\label{integraleq3}
F(u;\beta,\rho,p) &= p \frac{1}{\sqrt{2\pi\beta}(e^u-1)}
e^{-\frac{1}{2\beta}(\log(e^u-1)+\frac12\beta-\rho)^2} 
\\
&\qquad\qquad
+(1-p)e^{\beta-\rho}
\int_0^\infty \frac{dw}{\sqrt{2\pi\beta}}
e^{-\frac{1}{2\beta}(w - w_0(u))^2} F(w;\beta,\rho,p),
\nonumber
\end{align}
with $w_0(u) = \log(e^u-1) + \frac32 \beta - \rho$.
We solve this integral equation 
by iteration, as in the case of the perpetuity $X_\infty$. 

What is a reasonable range for the parameter $p$ of the geometric
mortality model? We will constrain this parameter by matching to the
Makeham model, which is one of the popular models for mortality 
rates \cite{Bowers}. Under this model the yearly mortality rate $\mu(a)$ 
is specified by the functional form \cite{DufresneAnnuity}
\begin{equation}
\mu(a) = A + B e^{\beta x}.
\end{equation}
The model parameters are \cite{Bowers}
\begin{equation}
A = 0.0007\,, \qquad B = 5\cdot 10^{-5}\,,\qquad \beta = 0.0921\,.
\end{equation}

We will match the geometric mortality model parameter $p$ to this 
realistic model. We consider two possible approaches.

i) Matching the expected life average, conditional on survival at age
$a_0$. This gives the equation
\begin{equation}
\mathbb{E}[a|a_0] = \int_{a_0}^\infty  t p_s(t) \mu(t) dt = 
a_0 + \frac{1}{p},
\end{equation}
where the conditional survival probability $p_s(a|a_0)$ is given by
\begin{equation}
p_s(a|a_0) = \exp\left( - \int_{a_0}^a ds \mu(s) \right).
\end{equation}
This method gives $p=0.06443$ for $a_0=65$.

ii) Matching the mortality rates at age $a_0$.  This gives the
relation
\begin{equation}
\mu(a_0) = p(1-p)^{a_0}.
\end{equation}
This method gives $p=0.02132$ for $a_0=65$. 

The two values of $p$ determined as above
span a range of realistic values, with i) being on the upper
side of the realistic estimates for the $p$ parameter, while ii) is
on the lower side of the range. A more precise matching is possible,
using the result noted in \cite{GSY} that any positive-definite
discrete distribution can be approximated arbitrarily close by an 
appropriate linear combination of geometric distributions. 
A similar result holds in continuos time \cite{Dufresne2007}, and 
states that any positive definite continuous
distribution can be approximated arbitrarily close by an appropriate
linear combination of exponential distributions. 

In order to study the effect of the geometric mortality on the
shape of the density of $X_N$, we show in Figure~\ref{Fig:4}
the results with $p=0.01$ and $p=0.1$ (black curves), comparing with the
results for the density of $X_\infty$ ($p=0$) (dashed curves).
The remaining parameters are $\beta=1$ (above) and $\beta=0.1$ 
(lower) and $\rho=0$.
As expected, the effect of turning on a non-zero 
mortality rate $p$ is to increase the density at low values of $X_N$ 
and decrease the contribution of larger values of $X_N$. 

In the small $\tau$ limit the distribution of $\tau X_N$ is expected
to approach the distribution of the time integral of the GBM up to
an exponentially distributed time $Y_{T_\lambda}$, as stated in (\ref{fNsmalltau}).
The asymptotic continuous time distribution is shown as the blue curves in
Figure~\ref{Fig:4}. The agreement with the continuous time result is seen to 
become better for smaller values of $\beta,p$, as predicted by equation 
(\ref{fNsmalltau}).
The plots in Figure~\ref{Fig:4} show that the approach to the
continuous time limit is qualitatively different in the right and left
tails of the distribution. While the right tail has a similar shape to that
of the continuous time distribution (shown as the blue curve), the
left tail displays very different qualitative behavior in the continuous
and discrete time. We note two main differences in the left tail behavior:

i) the discrete time density of $X_N$ vanishes as $X_N\to 0$, while 
the continuous time density of $Y_{T_\lambda}$ approaches a finite value 
at origin $\phi_\lambda(0;\beta,\rho,p)=p$, as shown in 
Proposition~\ref{prop:2app}.

ii) for certain parameter values, the discrete time density of $X_N$  
has a bimodal distribution, with a peak visible near small values of $X_N$.
This is due to the contribution to the density from events with killing at 
the first time of the annuity. In contrast, the density of the asymptotic
continuous time distribution of $Y_{T_\lambda}$ has a unimodal shape.

The average value of the annuity is known exactly from 
Section~\ref{Sec:PositiveMoments}
\begin{equation}
\mathbb{E}[X_N] = \frac{e^\rho}{1-(1-p)e^\rho}\,.
\end{equation}
We checked the quality of the numerical solution by computing this 
expectation by numerical quadrature and confirming good agreement
with the theoretical result. For $p=0.1, \rho=0$ with $n=200$ iterations
we get $\mathbb{E}[X_N]= 9.9998$ which is close to the theoretical result 
($\mathbb{E}[X_N]=10$).

Following \cite{DufresneAnnuity} we compute the shortfall probability
\begin{equation}
\mathbb{P}_{qK} = \mathbb{P}(X_N > (1+q)K) \,,
\end{equation}
with $K=\mathbb{E}[X_N]$ and $q=0, 0.5$. 
This measures the probability that an initial investment in the amount of
100\% and 150\% of the average annuity value, respectively, will not be 
sufficient to make the required payments on the annuity. The results are 
given in Table~\ref{Table:1}, for the same parameters $(\beta,p)$ as in 
Figure~\ref{Fig:4}.
In the last column we show the continuous time result for the shortfall 
probability given by 
\begin{eqnarray}\label{PqKcont}
\mathbb{P}_{qK}^{\tau=0} = \bar \Phi_\lambda\left(\frac12\sigma^2 K(1+q);\alpha,\beta\right),
\end{eqnarray}
with $\bar\Phi_\lambda$ given in (\ref{Philambdabar}).
The difference with the discrete time result shows the effect of the 
discrete time nature of the annuity $X_N$ on its distribution and tail
behavior. We note that, for the parameters considered, the discrete time
and continuous time shortfall probabilities are very similar, the difference
between them being below 1\%. The agreement improves for smaller $\beta,p$, as
expected from the asymptotic result of Proposition~\ref{prop:XNsmalltau}. 
We conclude that
the continuous time approximation is reasonably good for practical purposes
in the right tail  of the annuity distribution, although the left tail can have
very different behavior even for comparatively small values of $\beta,p$.


\begin{table}[t!]
\caption{\label{Table:1} 
Shortfall probabilities $\mathbb{P}_{qK}$ for a single-life annuity with 
discrete time payouts, for several choices of the model parameters. The last 
column shows the continuous time result given by equation (\ref{PqKcont}).}
\begin{center}
\begin{tabular}{cccccccc}
\hline
$\beta$ & $\rho$ & $p$ & $\mathbb{E}[X_N]$ & $q$ & $(1+q)K$ & 
          $\mathbb{P}_{qK}$ & $\mathbb{P}_{qK}^{\tau=0}$ \\
\hline
\hline
1 & 0 & 0.1  & 10  & 0   & 10  & 0.10852 & 0.10658 \\
1 & 0 & 0.1  & 10  & 0.5 & 15  & 0.07122 & 0.06849 \\
1 & 0 & 0.01 & 100 & 0   & 100 & 0.01781 & 0.01783 \\
1 & 0 & 0.01 & 100 & 0.5 & 150 & 0.01187 & 0.01183 \\
\hline
0.1 & 0 & 0.1  & 10  & 0   & 10  & 0.26821 & 0.27067 \\
0.1 & 0 & 0.1  & 10  & 0.5 & 15  & 0.15846 & 0.15899 \\
0.1 & 0 & 0.01 & 100 & 0   & 100 & 0.10625 & 0.10658 \\
0.1 & 0 & 0.01 & 100 & 0.5 & 150 & 0.06853 & 0.06849 \\
\hline
1.0 & -0.1 & 0.1  & 4.87398  & 0   & 4.87398  & 0.16415 & 0.17072 \\
1.0 & -0.1 & 0.1  & 4.87398  & 0.5 & 7.31098  & 0.10509 & 0.10605 \\
1.0 & -0.1 & 0.01 & 8.68275  & 0   & 8.68275  & 0.12334 & 0.13083 \\
1.0 & -0.1 & 0.01 & 8.68275 & 0.5  & 13.02412 & 0.08018 & 0.08321 \\
\hline
0.1 & -0.1 & 0.1  & 4.87398 & 0   & 4.87398  & 0.34969 & 0.37558 \\
0.1 & -0.1 & 0.1  & 4.87398 & 0.5 & 7.31098  & 0.17592 & 0.19044 \\
0.1 & -0.1 & 0.01 & 8.68275 & 0   & 8.68275  & 0.32828 & 0.35577 \\
0.1 & -0.1 & 0.01 & 8.68275 & 0.5 & 13.02412 & 0.15949 & 0.17039 \\
\hline
\end{tabular}
\end{center}
\end{table}

\subsection{Application to Pricing Asian Options}

We illustrate in this section the application of the methods discussed
in Section~\ref{ApplicationSection} for pricing Asian options in the 
Black-Scholes model with discrete time averaging. We also discuss the
conditions of applicability of the methods of approximating the distribution
of the finite sum of GBM with the infinite sum $X_\infty$ or integral
$Y_\infty$, as proposed in \cite{MP}.

Theorem~\ref{Thm22} gives a method for computing recursively the probability 
density function of $X_n$, the finite sum of GBM. This can be used to 
price Asian options with discrete time sampling in the BS model.
For a more efficient calculation of the integral we change variables as 
in (\ref{integraleq2}). 

Using this method we priced Asian options with discrete sampling under
the scenarios considered in Vecer's paper \cite{Vecer}: 
$\sigma=0.4, r=0.1, T=1, K=100$,
spot price $S_0=95,100,105$ and several choices for the number of time 
steps $n = T/\tau$ from 10 to 1000. 
Table~\ref{Table:2} shows the results of the calculation 
for the discounted call Asian option price $C(K,T) = e^{-rT} 
\mathbb{E}[(X_n/n - K)_+]$. The results are in good agreement with those
obtained by Vecer (Table B in \cite{Vecer}) using a PDE method proposed in 
\cite{VecerJCF}, 
by Curran \cite{Curran} using a method based on conditioning on the geometric 
average, and by Tavella-Randall \cite{TR} using the PDE method of Rogers
and Shi \cite{RogersShi}.

The quality of the numerical integration was checked by computing also the 
expectation $\mathbb{E}[X_n]$ and checking that it agrees with the exact 
theoretical result, given by
\begin{equation}
\mathbb{E}[X_{n}] = S_0 \sum_{k=1}^n e^{kr \tau} = 
S_0\frac{e^{nr\tau}-1}{1 - e^{-r\tau}},
\end{equation}
up to four significant decimal points. We also priced put options and 
checked that put-call parity is satisfied to four significant decimal 
points.
\begin{equation}
C(K,T) - P(K,T) = e^{-rT} \left(\frac{S_0}{rT}(e^{rT}-1) - K\right)\,.
\end{equation}

\begin{table}
\caption{\label{Table:2} 
Numerical results for Asian options with discrete sampling in the Black-Scholes
model using the method described in this paper, under the scenarios discussed
in text. The results are compared with those of Vecer \cite{Vecer}, Curran
\cite{Curran} and Tavella-Randall \cite{TR}.}
\begin{center}
\begin{tabular}{cccccc}
\hline
$n$ & $S_0$ & $C(K,T)$ & Vecer 
            & Curran   & Tavella-Randall \\
\hline
\hline
10 &  95 & 9.2239  & 9.2228   &  9.2197 &  9.2149 \\
10 & 100 & 12.0424  & 12.042   & 12.0390 & 12.0348 \\
10 & 105 & 15.2243  & 15.2234  & 15.2202 & 15.2168 \\
\hline
25 &  95 &  8.7086  & 8.708   & 8.7053 &  8.6974 \\
25 & 100 & 11.4910  & 11.4906 & 11.4881 & 11.4803 \\
25 & 105 & 14.6510  & 14.651  & 14.6483 & 14.6415 \\
\hline
50 &  95 &  8.5371  & 8.5367  & 8.5340 & 8.5383 \\
50 & 100 & 11.3070  & 11.3068 & 11.3043 & 11.2982 \\
50 & 105 & 14.4611  & 14.4601 & 14.4575 & 14.4519 \\
\hline
125 &  95 & 8.4347  & 8.4339 &  8.4314  & 8.4304 \\
125 & 100 & 11.1974  & 11.1967 & 11.1940 & 11.1929 \\
125 & 105 & 14.3459  & 14.3455 & 14.3430 & 14.3424 \\
\hline
250 &  95 &  8.4006 &  8.4001  &  8.3972 & 8.3972 \\
250 & 100 & 11.1607  & 11.1600  & 11.1572 & 11.1573 \\
250 & 105 & 14.3081   & 14.3073  & 14.3048 & 14.3054 \\
\hline
500 &  95 &  8.3831  & 8.3826  &   8.3801  & 8.3804 \\
500 & 100 & 11.1422  & 11.1416 &  11.1388  & 11.1392 \\
500 & 105 & 14.2887  & 14.2881 &  14.2857  & 14.2866 \\
\hline
1000 &  95 &  8.3718  & 8.3741 &   8.3715  & 8.3719 \\
1000 & 100 & 11.1301  & 11.1322 & 11.1296  & 11.1300 \\
1000 & 105 & 14.2754  & 14.2786 & 14.2762  & 14.2771 \\
\hline
\end{tabular}
\end{center}
\end{table}

Next we discuss the pricing of Asian options under the following 
approximations:
\begin{enumerate}

\item Approximating the distribution of $\tau X_n$ with that of 
$Y_\infty$, the infinite time integral of the GBM. This is the 
theoretical basis of the approximation 
of Milevsky and Posner \cite{MP}, who propose a moment
matching approach for modeling the distribution of $X_n$ with an
inverse Gamma distribution.

\item Approximating the distribution of $X_n$ with that of $X_\infty$,
the infinite sum of GBM sampled on discrete times.

\end{enumerate}

These approximations are expected to be most precise in the limit of a
very large number of sampling times $n\to \infty$. 
We note that these approximations can only be used if $r< 0$,
which is the condition for the existence of the limiting distributions for
$Y_\infty$ (Theorem~\ref{Thm:Dufresne}) and for $X_\infty$ 
(Proposition~\ref{prop:2}), and for the finiteness of the expectations
$\mathbb{E}[Y_\infty] < \infty, \mathbb{E}[X_\infty] < \infty$.

Allowing also for a continuously paid dividend yield $q$ this condition reads
$r-q<0$. In usual applications to equity options the difference $r-q$ is
positive, although it could become negative in low rates environments.
A natural case where negative drifts appear is for Asian options on FX rates. 
Denote $X_t$ a currency exchange rate, defined as the number of 
units of domestic currency corresponding to one unit of foreign currency. 
One of the simplest models for the dynamics of exchange rates is to assume 
that $X_t$ satisfies the diffusion $dX_t/X_t = (r_d - r_f) dt + \sigma_X dW_t$ 
in the domestic currency risk-neutral measure. The drift of $X_t$
is given by the difference between the domestic and foreign
interest rates, such that the condition of applicability of the approximations
mentioned above is $r_d - r_f < 0$.

Another condition for the applicability of the approximations (1),(2)
is that the expectations of $\tau X_n$ and $Y_\infty$ ($X_n$ and $X_\infty$) should
be sufficiently close. For (1) this condition reads $\tau \frac{e^{nr\tau}-1}
{1 - e^{-r\tau}} \sim - \frac{1}{r}$ which requires in addition 
$n \tau \gg 1$ and $|r\tau | \ll 1$.
In practice this requires that the option maturity $T$ be sufficiently long
such that $|rT| \gg 1$, and that the time step is sufficiently small.
On the other hand, the corresponding condition for the approximation (2) is
$\frac{e^{nr\tau}-1}{1 - e^{-r\tau}}
\sim  \frac{1}{e^{-r\tau}-1}$ which requires only $n \gg 1$ and $n\tau \gg 1$
but does not impose any constraints on the size of the time step $\tau$.

\section*{Acknowledgements}
We would like to thank an anonymous referee and the editor for helpful
suggestions and comments. Lingjiong Zhu gratefully acknowledges support 
from the National Science Foundation via the award NSF-DMS-1613164.


\appendix
\section{The Law of the Time Integral of the GBM Up to an Exponentially 
Distributed Time}
\label{app}

We summarize in this Appendix for the convenience of the reader
a few known results concerning the distribution
of the time integral of GBM up to an exponentially distributed random time.
These results have been derived in \cite{Yor0}. See also \cite{DufresneReview} 
for a survey of related results.

Let us define
\begin{equation}
Y_{T_\lambda} := \int_0^{T_\lambda} dt
e^{\sigma W_t + (m - \frac12\sigma^2) t}.
\end{equation}
the time integral of the GBM up to an exponentially distributed time 
$T_\lambda \sim \mathbf{Exp}(\lambda)$. 
The probability density function of $Y_{T_\lambda}$ is given by the 
following result.

\begin{proposition}\label{prop:1app}
The probability density function of $Y_{T_\lambda}$ is given by
\begin{equation}\label{Phiz}
\phi_\lambda(z;\sigma,m,\lambda) = 
\frac12 \sigma^2 \varphi\left(\frac12 \sigma^2 z; \alpha,\beta\right),
\end{equation}
where 
\begin{equation}\label{varphiy}
\varphi(y;\alpha,\beta)= 
\frac{\alpha \beta \Gamma(\alpha)}{\Gamma(\alpha+\beta+1)} y^{-\beta-1}
{}_1 F_1\left(\beta+1,\alpha+\beta+1;-\frac{1}{y}\right),
\end{equation}
which is the density of the random variable $Y:=B_{1,\alpha}/G_{\beta}$ with 
$B_{1,\alpha}$ and $G_\beta$ being independent random variables 
distributed as Beta and Gamma distributions. Their densities are
\begin{eqnarray}\label{BetaPDF}
\varphi_B(b) &=& \alpha (1-b)^{\alpha-1}\,, \quad b \in [0,1], \\
\label{GammaPDF}
\varphi_G(g) &=& \frac{1}{\Gamma(\beta)} g^{\beta-1} e^{-g}\,, \quad
g \in [0,\infty) \,.
\end{eqnarray}

The constants $\alpha,\beta$ are
\begin{eqnarray}\label{alphadef}
\alpha &=& \frac{1}{2\sigma^2}
\left( 2m - \sigma^2 + \sqrt{(2m-\sigma^2)^2+8\lambda\sigma^2} \right), 
\\
\label{betadef}
\beta &=& \frac{1}{2\sigma^2}
\left( - 2m + \sigma^2 + \sqrt{(2m-\sigma^2)^2+8\lambda\sigma^2} \right).
\end{eqnarray}
\end{proposition}
The result (\ref{varphiy}) agrees with the explicit result for this density 
in \cite{DufresneReview} (fourth equation from bottom on p.~12).

We note that for $\lambda>0$, the constant $\alpha$ is always strictly 
positive. This is a necessary condition for $B_{1,\alpha}$ to exist.

\begin{proposition}[Right tail behavior]
The right tail of the distribution $Y_{T_{\lambda}}$ is given by:
\begin{equation}\label{Philambdabar}
\mathbb{P}(Y_{T_\lambda} > z) = \int_z^\infty dx \phi_\lambda(x;\sigma,m,\lambda) 
= \int_{\frac12 \sigma^2 z}^\infty dy \varphi(y;\alpha,\beta) = 
\bar \Phi_\lambda\left(\frac12 \sigma^2 z;\alpha,\beta\right),
\end{equation}
where $\bar\Phi_\lambda(y;\alpha,\beta)$ is the complementary cumulative 
distribution of $Y:=B_{1,\alpha}/G_{\beta}$ with 
$B_{1,\alpha}$ and $G_\beta$ being independent random variables 
distributed as Beta and Gamma distributions,
which has the right tail asymptotics as $y\rightarrow\infty$:
\begin{equation}
\bar\Phi_\lambda(y;\alpha,\beta) = \int_y^\infty dw \varphi(w;\alpha,\beta) 
=\frac{\alpha \Gamma(\alpha)}{\Gamma(\alpha+\beta+1)} y^{-\beta} ( 1+ O(y^{-1})).
\end{equation}
\end{proposition}

We note that the exponent is the same as for the right tail asymptotics of 
the discrete time sum of GBM with geometric mortality derived in 
(\ref{mugeom}).  This can be seen by replacing $p=\lambda\tau$ in 
(\ref{mugeom})  and approximating $\log(1-p) \sim -p$. This gives 
for the exponent
\begin{equation}
\mu = \frac{1}{2\sigma^2} 
(-2m + \sigma^2 + \sqrt{(2m-\sigma^2)^2 + 8\sigma^2 \lambda}),
\end{equation}
which agrees precisely with $\beta$ defined in (\ref{betadef}).

\begin{proposition}[Left tail behavior]\label{prop:2app}
The left tail behavior of the density $\varphi(y;\alpha,\beta)$ is
\begin{equation}
\lim_{y\to 0} \varphi(y;\alpha,\beta) = \alpha\beta = \frac{2}{\sigma^2} \lambda \,.
\end{equation}
The probability density $\varphi(y;\alpha,\beta)$ approaches a 
non-vanishing constant near the origin $y\to 0$.
\end{proposition}

This tail behavior is different from the discrete time case, where we find 
that the density of $X_N$ always vanishes near the zero point. This behavior 
implies that all negative moments $\mathbb{E}[Y_{T_\lambda}^\theta]$
with $\theta \leq -1$ do not exist. On the other hand, for the discrete 
time case all the negative moments of $X_N$ are finite.

\textbf{Limiting case $\lambda=0$}. This corresponds to $T_\lambda \to \infty$, 
since the expectation of $T_\lambda$ under the exponential distribution is
\begin{equation}
\mathbb{E}[T_\lambda] = \frac{1}{\lambda}.
\end{equation}

We distinguish two cases:

i) $m < \frac12\sigma^2$. We have $(\alpha,\beta) = (0,1-\frac{2m}{\sigma^2})$;

ii) $ m > \frac12 \sigma^2$. This gives a negative $\alpha$, which is 
meaningless since the $\mbox{Beta}(1,\alpha)$ distribution is defined only for 
$\alpha \geq 0$. 

The confluent hypergeometric function in case (i) can be expressed in closed 
form using the identity ${}_1 F_1(b,b,z) = e^z$, which holds for any positive 
integer $b$. We have
\begin{equation}
\varphi\left(y;0,1-\frac{2m}{\sigma^2}\right)=  \frac{1}{\Gamma(1+\frac{2m}{\sigma^2})}
y^{-2+\frac{2m}{\sigma^2}} e^{-1/y}.
\end{equation}
Using (\ref{Phiz}) we get the density function of the infinite time 
integral of the GBM 
\begin{eqnarray}
\lim_{\lambda\to 0}\phi_\lambda(z;\sigma,m,0) &=& 
\frac12\sigma^2 \frac{1}{\Gamma(1+\frac{2m}{\sigma^2})}
\left(\frac12 \sigma^2 z\right)^{-2+\frac{2m}{\sigma^2}} \exp\left(- \frac{2}{\sigma^2 z} \right) \\
&=& (2/\sigma^2)^{1 - \frac{2m}{\sigma^2}} \frac{1}{(z^2)^{1-\frac{m}{\sigma^2}}} \frac{1}{\Gamma(1+\frac{2m}{\sigma^2})}
\exp\left(- \frac{2}{\sigma^2 z} \right).\nonumber
\end{eqnarray}
This agrees with $\phi_\infty(z;\sigma,m)$ in equation (\ref{invGamma}).

\begin{proof}[Proof of Proposition \ref{prop:1app}]

\textbf{Step 1}. Use a time change to relate $Y_{T_\lambda}$ to a certain 
integral for which we know the distribution from Yor's paper \cite{Yor0}. 
This is 
\begin{equation}
Y_{t_\lambda}^{(\mu)} = \int_0^{t_\lambda} e^{2\mu s + 2W_s} ds,
\end{equation}
where $t_\lambda \sim \mathbf{Exp}(\lambda)$. It is known \cite{Yor0} 
that the distribution of this integral is
\begin{equation}
2Y_{t_\lambda}^{(\mu)} = \frac{B_{1,\hat \alpha}}{G_{\hat\beta}},
\end{equation}
with $\hat\alpha = \frac12\mu +\frac12 \sqrt{\mu^2+2\lambda}$ and 
$\hat\beta = -\frac12\mu + \frac12 \sqrt{\mu^2+2\lambda}$.

It is easy to see, using a time change, that we have
\begin{equation}
Y_{T_\lambda} = \frac{4}{\sigma^2} 
   Y_{\frac14 \sigma^2 T_\lambda}^{(\frac{2m}{\sigma^2}-1)}.
\end{equation}
We have $\frac14 \sigma^2 T_\lambda \sim 
\mathbf{Exp}(\frac{4}{\sigma^2}\lambda)$. 
Substituting $\lambda \to \frac{4}{\sigma^2}\lambda$ and 
$\mu \to \frac{2m}{\sigma^2}-1$ into the expressions for 
$\hat\alpha,\hat\beta$ we get the results (\ref{alphadef}) and (\ref{betadef}) 
for the parameters $\alpha,\beta$.

\textbf{Step 2.} 
Use the result \cite{Yor0} 
\begin{equation}
Y_{T_\lambda} = \frac{2}{\sigma^2} \frac{B_{1,\alpha}}{G_\beta}
\end{equation}
in distribution. Denoting $y=b/g$ the ratio of the two independent random 
variables, it is easy to show by explicit calculation that its pdf is given 
by (\ref{varphiy}). Taking into account also the factor $2/\sigma^2$ we get 
the final result (\ref{Phiz}).
\end{proof}

\begin{proof}[Proof of Proposition \ref{prop:2app}]
We prove here the left tail asymptotics.
This follows from the asymptotic expansion for the confluent 
hypergeometric function of large negative argument:
\begin{equation}
{}_1 F_1\left(\beta+1,\alpha+\beta+1,-\frac{1}{y}\right) \sim 
\frac{\Gamma(\alpha+\beta+1)}{\Gamma(\alpha)} y^{\beta+1}\,, \quad 
\text{as $y \to 0$}.
\end{equation}
This is obtained from the asymptotics for large positive argument
\begin{equation}
{}_1 F_1(a,b,z) \sim \frac{\Gamma(b)}{\Gamma(a)} e^{z} z^{a-b}\,,\qquad
\text{as $z\to \infty$},
\end{equation}
together with the Kummer transformation relation
\begin{equation}
{}_1 F_1(a,b,z) = e^z {}_1 F_1(b-a,b,-z) \,.
\end{equation}
\end{proof}


\section{Proofs}

\label{App:proofs}

\begin{proof}[Proof of Theorem~\ref{InftyThm}]
(i) We have
\begin{align}
&\left|\tau\sum_{i=1}^{N}
e^{\sigma W_{t_{i-1}}+(m-\frac{1}{2}\sigma^{2})t_{i-1}}
-\int_{0}^{T}e^{\sigma W_{t}+(m-\frac{1}{2}\sigma^{2})t}dt\right|
\\
&\leq\sum_{i=1}^{N}
\left|\tau e^{\sigma W_{t_{i-1}}+(m-\frac{1}{2}\sigma^{2})t_{i-1}} -
\int_{t_{i-1}}^{t_{i}} e^{\sigma W_{t}+(m-\frac{1}{2}\sigma^{2})t}dt\right|
\nonumber
\\
&\leq\tau\sum_{i=1}^{N}  \max_{t_{i-1}\leq t\leq t_{i}}
\left|e^{\sigma W_{t}+(m-\frac{1}{2}\sigma^{2})t}
-e^{\sigma W_{t_{i-1}}+(m-\frac{1}{2}\sigma^{2}) t_{i-1}}\right|.
\nonumber
\end{align}

Add and subtract in each term $e^{\sigma W_t + (m-\frac12 \sigma^2) t_{i-1}}$
and use the triangle inequality. Each term in the sum is bounded from above
as
\begin{equation}
\tau \max_{t_{i-1}\leq t\leq t_{i}}
\left|e^{\sigma W_{t}+(m-\frac{1}{2}\sigma^{2})t}
-e^{\sigma W_{t_{i-1}}+(m-\frac{1}{2}\sigma^{2}) t_{i-1}}\right|
\leq T_1^{(i)} + T_2^{(i)},
\end{equation}
where we denoted
\begin{eqnarray}
T_1^{(i)} &=& 
\tau \max_{t_{i-1}\leq t\leq t_{i}}
\left| e^{\sigma W_{t}} - e^{\sigma W_{t_{i-1}}} \right|
e^{(m-\frac{1}{2}\sigma^{2}) t_{i-1}}, \\
T_2^{(i)} &=& 
\tau \max_{t_{i-1}\leq t\leq t_{i}}
e^{\sigma W_{t}}
\left| e^{(m-\frac{1}{2}\sigma^{2}) t} -  
       e^{(m-\frac{1}{2}\sigma^{2}) t_{i-1}} \right|.
\end{eqnarray}

We would like to show that the following sum converges to zero as $\tau \to 0$
\begin{equation}
\lim_{\tau \to 0, \tau N=T}
\sum_{i=1}^N \mathbb{E}\left[T_1^{(i)} + T_2^{(i)} \right] \to 0.
\end{equation}

We bound each term in turn. 
\begin{align}
\mathbb{E}\left[T_1^{(i)}\right] 
&=
\tau \mathbb{E}\left[ e^{\sigma W_{t_{i-1}} +(m-\frac12\sigma^2) t_{i-1}}
\max_{t_{i-1} \leq t \leq t_{i}} \left| e^{\sigma (W_{t}-W_{t_{i-1}})} - 1 \right|\right] 
\\
&=\tau e^{mt_{i-1}}\mathbb{E}\left[\max_{t_{i-1} \leq t \leq t_{i}}\left| e^{\sigma (W_{t}-W_{t_{i-1}})} - 1 \right|\right] 
\nonumber
\\
&\leq
\tau e^{m t_{i-1}} \mathbb{E}\left[\max_{0\leq t \leq \tau}
|e^{\sigma W_t} - 1 |\right] \nonumber 
\\
&\leq\tau e^{m t_{i-1}} \mathbb{E}\left[e^{\sigma\max_{0\leq t\leq\tau}W_{t}}-e^{\sigma\min_{0\leq t\leq\tau}W_{t}}\right] \nonumber 
\\
&= \tau e^{m t_{i-1}} \left(\mathbb{E}\left[e^{\sigma|W_{\tau}|}\right]
-\mathbb{E}\left[e^{-\sigma|W_{\tau}|}\right]\right), \nonumber 
\end{align}
where we used the reflection principle for Brownian motions.
Summing over $i$ we have
\begin{align}\label{205}
\sum_{i=1}^N \mathbb{E}\left[T_1^{(i)}\right] 
&\leq \tau 
\frac{e^{m\tau N} - 1}{e^{m\tau} - 1}
\left(\mathbb{E}\left[e^{\sigma|W_{\tau}|}\right]-\mathbb{E}\left[e^{-\sigma|W_{\tau}|}\right]\right)
\\
&=
\tau 
 \frac{e^{m T} - 1}{e^{m\tau} - 1} e^{\frac{1}{2}\sigma^{2}\tau}\cdot 2\left[\Phi(\sigma\sqrt{\tau})-\Phi(-\sigma\sqrt{\tau})\right]
\rightarrow 0 ,
\nonumber
\end{align}
as $\tau \to 0$, where $\Phi(x):=\int_{-\infty}^{x}\frac{1}{\sqrt{2\pi}}e^{-y^{2}/2}dy$ is the cumulative
distribution function of a standard normal random variable.

The second term is bounded in a similar way. 
\begin{align}
\mathbb{E}\left[T_2^{(i)}\right] 
& 
\leq \tau \mathbb{E}\left[\max_{t_{i-1} < t < t_i}
e^{\sigma (W_t - W_{t_{i-1}})}
e^{\sigma W_{t_{i-1}} + (m-\frac12\sigma^2) t_{i-1}} \right]
\\
&\qquad\qquad\cdot 
\max_{t_{i-1} < t < t_i} \left|e^{(m-\frac12\sigma^2)(t - t_{i-1})} - 1 \right| \,. \nonumber
\end{align}
The two factors in the expectation are again independent, so the expectation
factors into their expectations. The supremum over the last factor depends
on the sign of $m-\frac12\sigma^2$. For $m-\frac12 \sigma^2>0$ this is 
$e^{(m-\frac12\sigma^2)\tau}-1$ and for $m-\frac12 \sigma^2<0$ this is 
$1-e^{(m-\frac12\sigma^2)\tau}$. 
This gives
\begin{align}\label{207}
\mathbb{E}\left[T_2^{(i)}\right] 
&\leq \tau \mathbb{E}\left[\max_{t_{i-1} < t < t_i}
e^{\sigma (W_t - W_{t_{i-1}})} \right] e^{m t_{i-1}}
\mbox{sgn}(m-\frac12\sigma^2)
\left(e^{(m-\frac12\sigma^2)\tau} - 1 \right) 
\\
&\leq 2 \Phi(\sigma\sqrt{\tau}) \tau 
e^{\frac12 \sigma^2\tau} e^{m t_{i-1}}
\mbox{sgn}(m-\frac12\sigma^2)\left(e^{(m-\frac12\sigma^2)\tau} - 1 \right).
\nonumber \\
& = 2 \Phi(\sigma\sqrt{\tau}) \tau  e^{m t_{i-1}}
\left| e^{m\tau} - e^{\frac12\sigma^2\tau} \right|.
\nonumber
\end{align}
This is summed over $i=1,N$ as previously, and the result is finite and 
goes to zero as $\tau \to 0$.

(ii) The case of the infinite maturity $T\to \infty$ is treated analogously,
except that one requires $m<0$ in order to ensure the convergence of the sums
$\sum_{i=1}^\infty T_{1}^{(i)}$ and $\sum_{i=1}^\infty T_{2}^{(i)}$. 
\end{proof}


\begin{proof}[Proof of Theorem~\ref{ThmExponentialT}]
{\bf Step 1.} We start by proving the limit
\begin{equation}\label{eq1}
\lim_{\tau \to 0} \left|
\tau\sum_{i=1}^{N}e^{\sigma W_{t_{i-1}}+(m-\frac{1}{2}\sigma^{2})t_{i-1}}
- \int_{0}^{N\tau}e^{\sigma W_{t}+(m-\frac{1}{2}\sigma^{2})t}dt \right|
= 0
\end{equation}
in probability as $\tau\rightarrow 0$. Here $N$ is a geometrically 
distributed random variable with parameter $p$.

Consider the expectation
\begin{eqnarray}
\Delta_\tau := \mathbb{E}\left[ \left| \tau\sum_{i=1}^{N}
   e^{\sigma W_{t_{i-1}}+(m-\frac{1}{2}\sigma^{2})t_{i-1}}
- \int_{0}^{N\tau}e^{\sigma W_{t}+(m-\frac{1}{2}\sigma^{2})t}dt
\right| \right] \,.
\end{eqnarray}
In order to prove (\ref{eq1}) it is sufficient to show $\lim_{\tau\to 0}
\Delta_\tau = 0$.

This expectation is written as
\begin{eqnarray}\label{eq5}
\Delta_\tau  = \sum_{k=1}^\infty p(1-p)^{k-1} 
\mathbb{E}\left[\left|\tau\sum_{i=1}^k
   e^{\sigma W_{t_{i-1}}+(m-\frac{1}{2}\sigma^{2})t_{i-1}}
- \int_{0}^{k\tau}e^{\sigma W_{t}+(m-\frac{1}{2}\sigma^{2})t}dt\right|\right] \,.
\end{eqnarray}

The expectations with fixed $k\in \mathbb{N}$ are bounded from above 
as shown in the proof of Theorem~\ref{InftyThm}. We have
\begin{eqnarray}\label{Deltafixedk}
&& \mathbb{E}\left[ \left| \tau\sum_{i=1}^k
   e^{\sigma W_{t_{i-1}}+(m-\frac{1}{2}\sigma^{2})t_{i-1}}
- \int_{0}^{k\tau}e^{\sigma W_{t}+(m-\frac{1}{2}\sigma^{2})t}dt\right|\right]  
    \leq \sum_{i=1}^k \mathbb{E}[T_1^{(i)} + T_2^{(i)}] \\
&& \leq (e^{m k\tau}-1) R_\tau \nonumber
\end{eqnarray}
where we have from equation (\ref{205})
\begin{eqnarray}
\sum_{i=1}^k \mathbb{E}[T_1^{(i)}] \leq 
\tau \frac{e^{mk\tau}-1}{e^{m\tau}-1} e^{\frac12\sigma^2\tau} 2[\Phi(\sigma\sqrt{\tau}) - \Phi(-\sigma\sqrt{\tau})]
\end{eqnarray}
and from equation (\ref{207})
\begin{eqnarray}
\sum_{i=1}^k \mathbb{E}[T_2^{(i)}] \leq 
\tau \frac{e^{mk\tau}-1}{e^{m\tau}-1} 
\left| e^{m\tau}- e^{\frac12\sigma^2\tau} \right| 2 \Phi(\sigma\sqrt{\tau})
\end{eqnarray}

We combined these two inequalities in the last line of (\ref{Deltafixedk}) 
by introducing
\begin{eqnarray}
R_\tau &:=& \tau \frac{1}{e^{m\tau}-1} e^{\frac12 \sigma^2\tau} 2[\Phi(\sigma\sqrt{\tau})-\Phi(-\sigma\sqrt{\tau})] \\
&+& \tau \frac{1}{e^{m\tau}-1} 
\left| e^{m\tau}-e^{\frac12 \sigma^2\tau}\right|
2\Phi(\sigma\sqrt{\tau}) \,.\nonumber
\end{eqnarray}

Substituting (\ref{Deltafixedk}) into (\ref{eq5}) we have
\begin{eqnarray}
 \Delta_\tau  \leq R_\tau \sum_{k=1}^\infty p(1-p)^{k-1} (e^{km\tau}-1)
 = \left( \frac{p e^{m\tau}}{1-e^{m\tau}(1-p)} - 1\right) R_\tau \,.
\end{eqnarray}
The sum over $k$ converges provided that $(1-p) e^{m\tau} < 1$. A sufficient
condition for this to hold for any $\tau >0$ is $m < \lambda$.

As $\tau\to 0$, we have
\begin{eqnarray}
\lim_{\tau\to 0} \left(\frac{p e^{m\tau}}{1-e^{m\tau}(1-p)} - 1 \right) = \frac{m}{\lambda-m}
\end{eqnarray}
and 
\begin{eqnarray}
\lim_{\tau \to 0} R_\tau = 0
\end{eqnarray}
Combining these results gives $\lim_{\tau \to 0} \Delta_\tau = 0$
which completes the proof of the result stated.

{\bf Step 2.}
In the next step we prove
\begin{equation}
\int_{0}^{N\tau}e^{\sigma W_{t}+(m-\frac{1}{2}\sigma^{2})t}dt
\rightarrow\int_{0}^{T_{\lambda}}e^{\sigma W_{t}+(m-\frac{1}{2}\sigma^{2})t}dt,
\end{equation}
in distribution as $\tau\rightarrow 0$.
For any $x>0$, 
\begin{align}
&\mathbb{P}\left(\int_{0}^{N\tau}e^{\sigma W_{t}+(m-\frac{1}{2}\sigma^{2})t}dt\leq x\right)
\\
&=\sum_{k=1}^{\infty}\mathbb{P}\left(\int_{0}^{k\tau}e^{\sigma W_{t}+(m-\frac{1}{2}\sigma^{2})t}dt\leq x\right)
(1-\lambda\tau)^{k-1}\lambda\tau,
\nonumber
\\
&=\frac{1}{1-\lambda\tau}\sum_{k=1}^{\infty}\mathbb{P}\left(\int_{0}^{k\tau}e^{\sigma W_{t}+(m-\frac{1}{2}\sigma^{2})t}dt\leq x\right)
\left((1-\lambda\tau)^{\frac{1}{\tau}}\right)^{k\tau}\lambda\tau
\nonumber
\\
&\rightarrow\int_{0}^{\infty}\mathbb{P}\left(\int_{0}^{u}e^{\sigma W_{t}+(m-\frac{1}{2}\sigma^{2})t}dt\leq x\right)
\lambda e^{-\lambda u}du
\nonumber
\\
&=\mathbb{P}\left(\int_{0}^{T_{\lambda}}
  e^{\sigma W_{t}+(m-\frac{1}{2}\sigma^{2})t}dt\leq x\right),
\nonumber
\end{align}
as $\tau\rightarrow 0$. Hence, we proved the desired result.
\end{proof}


\end{document}